\documentclass{article}

\input glyphtounicode
\usepackage[T1]{fontenc}
\usepackage[utf8]{inputenc}

\usepackage{ifdraft}

\usepackage[english]{babel}

\usepackage{mathtools}
\usepackage{enumitem}
\usepackage[amsmath,thmmarks,hyperref]{ntheorem}

\usepackage{tikz}

\usepackage{xcolor}
\definecolor{hypercolor}{rgb}{0,0.2,0.7}
\usepackage[breaklinks,final]{hyperref}

\usepackage[numbers,sort&compress]{natbib}
\newif\iffancyfont%
\fancyfontfalse%

\IfFileExists{MinionPro.sty}{\IfFileExists{MyriadPro.sty}{\fancyfonttrue}{}}{}

\iffancyfont%
  \usepackage[lf,medfamily,italicgreek,openg]{MinionPro}
  \usepackage[lf,medfamily,onlytext,scale=0.96]{MyriadPro}

  \renewcommand{\dotsb}{{\mathinner{\cdotp\cdotp\cdotp}}}

  \DeclareRobustCommand{\defn}{\mathrel{{\vdotdot}{\equal}}}

  \usepackage[toc,eqno,bib]{tabfigures}
  \newtagform{tabular}[\figureversion{tabular}]{(}{)}
  \usetagform{tabular}
\else%
  \usepackage[full]{textcomp}
  \usepackage{newpxtext}
  \usepackage[bigdelims,varg]{newpxmath}

  \let\mathbb\vmathbb

  \newcommand{\defn}{\coloneqq}
  
\fi%

\usepackage{dsfont}

\usepackage[scr=rsfso,cal=cm]{mathalfa}

\usepackage[final]{microtype}

\usepackage[sf,bf,small]{titlesec}
\usepackage[labelfont={sf,bf},labelsep=period]{caption}

\setlist[1]{labelindent=\parindent}
\setlist[description]{font=\sffamily\bfseries,align=right,labelsep=1em}

\numberwithin{equation}{section}

\setlist[enumerate,1]{label=(\roman*),font=\normalfont}

\makeatletter
\newcounter{and}
\newdimen{\instindent}
\newcommand{\institute}[1]{\newcommand{\@institute}{#1}}
\newcommand{\inst}[1]{\unskip\smash{$^{#1}$}\setcounter{and}{1}\ignorespaces}
\newcommand{\email}[1]{\href{mailto:#1}{#1}}
\renewcommand{\maketitle}{
  {
    \raggedright%
    \LARGE%
    \noindent%
    \bfseries%
    \sffamily%
    \@title%
    \par
  }

  \vspace{1.5\baselineskip}

  {
    \raggedright%
    \renewcommand{\and}{\unskip, \ignorespaces}%
    \noindent\ignorespaces\@author\par
  }

  \vspace{0.5\baselineskip}

  {
    \small%
    \parindent=0pt%
    \parskip=0pt%
    \setcounter{and}{1}%
    \renewcommand{\and}{%
      \par\stepcounter{and}%
      \hangindent\instindent%
      \noindent%
      \hbox to \instindent{\hss\smash{$^{\theand}$\enspace}}\ignorespaces%
    }%
    \setbox0=\vbox{\@institute}%
    \ifnum\value{and}>9\relax\setbox0=\hbox{$^{88}$\enspace}%
    \else\setbox0=\hbox{$^{8}$\enspace}\fi%
    \instindent=\wd0\relax%
    \ifnum\value{and}=1\relax%
    \else%
      \setcounter{and}{1}%
      \hangindent\instindent%
      \noindent%
      \hbox to \instindent{\hss\smash{$^{\theand}$}\enspace}\ignorespaces%
    \fi%
    \ignorespaces%
    \@institute\par
  }
}
\makeatother

\renewenvironment{abstract}{
  \addvspace{1.5\baselineskip}%
  \topsep=0pt\partopsep=0pt%
  \trivlist\item[\hspace{\labelsep}\bfseries\sffamily Abstract.]
}{}

\newenvironment{acknowledgments}{
  \addvspace{1.5\baselineskip}%
  \topsep=0pt\partopsep=0pt%
  \trivlist\item[\hspace{\labelsep}\bfseries\sffamily Acknowledgments.]
}{}

\usepackage[obeyDraft]{todonotes}

\ifdraft{
  \usepackage[notref,notcite]{showkeys}

  \usepackage{filemod}
  \usepackage{fancyhdr}
  \pagestyle{fancy}
  \fancyhf{}
  
  \fancyhead[C]{Draft version. Last modified: \filemodprint{\jobname}}
  \fancyfoot[C]{\thepage}
}{}

\newcommand{\ie}{\textit{i.e.}}
\newcommand{\eg}{\textit{e.g.}}
\newcommand{\viz}{\textit{viz.}}
\newcommand{\cf}{\textit{cf.}}

\newcommand{\dif}{\mathrm{d}}
\newcommand{\e}{\mathrm{e}}
\newcommand{\im}{\mathrm{i}}

\newcommand{\field}[1][K]{{\mathds{#1}}}

\newcommand{\RR}{{\field[R]}}
\newcommand{\CC}{{\field[C]}}

\renewcommand{\Re}{\operatorname{Re}}
\renewcommand{\Im}{\operatorname{Im}}

\DeclareMathOperator{\supp}{supp}
\DeclareMathOperator{\spec}{sp}
\DeclareMathOperator{\reso}{rs}

\DeclareMathOperator{\dom}{dom}
\DeclareMathOperator{\Dom}{Dom}
\DeclareMathOperator{\Ran}{Ran}
\DeclareMathOperator{\Ker}{Ker}

\DeclareMathOperator*{\slim}{s-lim}
\DeclareMathOperator*{\wlim}{w-lim}
\DeclareMathOperator{\sgn}{sgn}

\newcommand{\comp}{\mathrm{c}}
\newcommand{\loc}{\mathrm{loc}}
\newcommand{\bnd}{\mathrm{b}}

\newcommand{\conj}[1]{\overline{#1}}
\DeclarePairedDelimiter{\abs}{\lvert}{\rvert}
\DeclarePairedDelimiter{\norm}{\lVert}{\rVert}
\DeclarePairedDelimiter{\jnorm}{\langle}{\rangle}
\DeclarePairedDelimiter{\floor}{\lfloor}{\rfloor}

\newcommand{\one}{\mathds{1}}

\newcommand{\mx}{\mathrm{mx}}
\newcommand{\mn}{\mathrm{mn}}

\newcommand{\Feyn}{\mathrm{F}}
\newcommand{\aFeyn}{{\overline{\mathrm{F}}}}
\newcommand{\FeynFeyn}{{\Feyn/\mkern1mu\aFeyn}}
\newcommand{\retadv}{{\vee\mkern-2mu/\mkern-2mu\wedge}}
\newcommand{\PJ}{\mathrm{PJ}}
\newcommand{\en}{\mathrm{en}}
\newcommand{\dyn}{\mathrm{dyn}}

\newcommand{\Hen}[1][]{\mathcal{H}_{\en\ifx\\#1\\\else,#1\fi}}
\newcommand{\Hdyn}[1][]{\mathcal{H}_{\dyn\ifx\\#1\\\else,#1\fi}}
\newcommand{\Hne}[1][]{\mathcal{H}^{\mathrlap{*}}{}_{\en\ifx\\#1\\\else,#1\fi}}

\newcommand{\KG}{K}

\newcommand{\twomat}[4]{\begin{pmatrix} #1 & #2 \\ #3 & #4 \end{pmatrix}}

\makeatletter
\newcommand{\aotimes}{{\mathop{\otimes}\limits^{
  \vbox to .15ex {\kern-2\ex@\hbox{\tiny alg}\vss}}}}
\makeatother

\newcommand{\hyp}{-\penalty0\hskip0pt\relax}

\makeatletter
\newcommand{\manlabel}[2]{#2\def\@currentlabel{#2}\label{#1}}
\makeatother

\makeatletter
\newtheoremstyle{nonumberplainnoparens}%
  {\item[\theorem@headerfont\hskip\labelsep ##1\theorem@separator]}%
  {\item[\theorem@headerfont\hskip\labelsep ##1 ##3\theorem@separator]}
\makeatother

\theoremnumbering{arabic}
\qedsymbol{\ensuremath{\quad\Box}}

\theoremstyle{plain}

\theorembodyfont{\itshape}
\theoremheaderfont{\normalfont\sffamily\bfseries}
\theoremseparator{.}

\newtheorem{theorem}{Theorem}[section]
\newtheorem{proposition}[theorem]{Proposition}
\newtheorem{lemma}[theorem]{Lemma}
\newtheorem{corollary}[theorem]{Corollary}

\theorembodyfont{\normalfont}

\newtheorem{definition}[theorem]{Definition}

\newtheorem{remark}[theorem]{Remark}

\newtheorem{assumption}{Assumption}

\theoremstyle{nonumberplainnoparens}
\theoremsymbol{\ensuremath{\quad\Box}}
\newtheorem{proof}{Proof}

\title{An Evolution Equation Approach to the Klein--Gordon Operator on Curved Spacetime}

\author{
  Jan Dereziński\inst{1}
  \and
  Daniel Siemssen\inst{1,2}
}

\institute{
  Department of Mathematical Methods in Physics, Faculty of Physics, University of Warsaw, Pasteura 5, 02-093, Warszawa, Poland.
  E-mail:~\email{jan.derezinski@fuw.edu.pl}.
  \and
  Department of Mathematics and Informatics, University of Wuppertal, Gaußstraße 20, 42119 Wuppertal, Germany.
  E-mail:~\email{siemssen@uni-wuppertal.de}
}

\hypersetup{
  unicode,
  colorlinks,
  linkcolor=hypercolor,
  urlcolor=hypercolor,
  citecolor=hypercolor,
  bookmarksnumbered,
  bookmarksdepth=2,
  pdfauthor={Jan Dereziński, Daniel Siemssen},
  pdftitle={An Evolution Equation Approach to the Klein–Gordon Operator on Curved Spacetime}
}

\begin{document}

\maketitle

\begin{abstract}
  We develop a theory of the Klein--Gordon equation on curved spacetimes.
  Our main tool is the method of (non-autonomous) evolution equations on Hilbert spaces.
  This approach allows us to treat low regularity of the metric, of the electromagnetic potential and of the scalar potential.
  Our main goal is a construction of various kinds of propagators needed in quantum field theory.
\end{abstract}



\section{Introduction}

We consider the \emph{Klein--Gordon operator on a Lorentzian manifold $(M, g)$ minimally coupled to an electromagnetic potential $A$ and with a scalar potential $Y$}.
In local coordinates it can be written as
\begin{equation}\label{eq:klein-gordon}
  \KG \defn \Box_A + Y = \abs{g}^{-\frac12} (D_\mu - A_\mu) \abs{g}^\frac12 g^{\mu\nu} (D_\nu - A_\nu) + Y,
\end{equation}
where $\abs{g} = \abs{\det [g_{\mu\nu}]}$ and $D_\mu = -\im\partial_\mu$.
As in our recent work~\cite{derezinski-siemssen}, we are interested in inverses and bisolutions of the Klein--Gordon operator $\KG$.

Heuristically, they are defined as follows:
\begin{itemize}
  \item an operator $G$ is a \emph{bisolution} of $\KG$ if it satisfies
    \begin{equation*}
      \KG G = 0
      \quad\text{and}\quad
      G \KG = 0.
    \end{equation*}
  \item an operator $G$ is an \emph{inverse} of $\KG$ if it satisfies
    \begin{equation*}
      \KG G = \one
      \quad\text{and}\quad
      G \KG = \one.
    \end{equation*}
\end{itemize}
To make these statements rigorous, one needs to specify the spaces between which these operators act, making sure that the composition of $\KG$ and $G$  is well-defined.
Often, $G$ can be understood as an operator from $C^\infty_\comp(M)$ to $C^\infty(M)$.

The Klein--Gordon operator has several distinguished inverses and bisolutions.
They are known by many names, \eg, ``propagator'' or ``two-point function''.
Inverses are often also called ``Green's functions''.

The most well-known propagators are probably the \emph{forward (retarded) propagator} $G^\vee$ and the \emph{backward (advanced) propagator} $G^\wedge$.
Their difference $G^\PJ \defn G^\vee - G^\wedge$ is sometimes called the \emph{Pauli--Jordan propagator}, which is the name we use.
In the literature one can also find other names, such as ``commutator function'' or ``causal propagator''.%
\footnote{%
  We try to use as much as possible the terminology from classic textbooks on quantum field theory.
  For instance, ``Pauli--Jordan function'' is the name used for $G^\PJ$ already
  in Bogoliubov--Shirkov \cite{BS}.
  The same authors call $G^\Feyn$ the ``causal Green's function'', since
  the choice of $G^\Feyn$ for the evaluation of Feynman diagrams expresses  causality in quantum field theory.
  Therefore, using the name ``causal propagator'' for $G^\PJ$ clashes with the traditional terminology and, we believe, should be discouraged.
}
These three propagators are important in the Cauchy problem of the classical theory.
Therefore, we will call them jointly \emph{classical propagators}.
It is well-known that on globally hyperbolic spacetimes the classical propagators exist and are unique.

In quantum field theory, one needs also other propagators: two inverses, the \emph{Feynman propagator}~$G^\Feyn$ and the \emph{anti-Feynman propagator}~$G^\aFeyn$, as well as the \emph{positive} and \emph{negative frequency bisolutions}~$G^{(\pm)}$.
We will call them jointly \emph{non-classical propagators}.
A positive frequency bisolution yields the two-point function of a vacuum state -- a pure quasi-free state whose Gelfand--Naimark--Segal (GNS) representation yields a Hilbert space for the quantum field theory.
The integral kernel of the Feynman propagator coincides with the expectation value of time-ordered products of quantum fields.
It is used to evaluate Feynman diagrams.

The analysis of the Klein-Gordon equation is especially simple if the spacetime is stationary and the Hamiltonian is positive.
On the mathematical side, if in addition the Hamiltonian is bounded away from zero (the ``positive mass case''), we have a natural Hilbert space structure for the Cauchy data.
The most obvious choice is the so-called \emph{energy Hilbert space}.
It is also natural to consider a whole scale of Hilbert spaces, which includes the energy space.
The generator of the dynamics is self-adjoint on all of these spaces.
Thus the functional analytic setting for stationary spacetimes in the ``positive mass case'' is rather clean and simple.
If we assume that the Hamiltonian is only positive, without a positive lower bound, (the ``zero mass case''), then the functional-analytic setup becomes slightly more technically involved, but the general picture remains the same.

On the physical side, on a stationary spacetime with a positive Hamiltonian, it is clear how to define the non-classical propagators.
The positive and negative frequency bisolutions, as well as the Feynman and anti-Feynman propagators, are constructed from the spectral projections of the generator of the dynamics.
These constructions, at least implicitly, can be found in various works devoted to quantum field theory on curved spacetimes.
In a systematic way the static case has been worked out recently in~\cite{derezinski-siemssen}, see also Chap.~18 of~\cite{derezinski-gerard}.
\cite{derezinski-siemssen} assumed in addition the ``positive mass condition''.
The results of~\cite{derezinski-siemssen} can be easily generalized to stationary spacetimes (using \eg\ the stationary special case of Sects.~\ref{sub:1+3_splitting} and~\ref{sub:KG_operator} as a starting point).

The positivity of the Hamiltonian plays an important role in the construction of non-classical propagators.
This is related to the fact that non-positive Hamiltonians lead to problems in quantum field theory, which are often collectively called the \emph{Klein paradox}.
The original paper by Klein involved fermions and the Dirac equation with a large step potential causing spontaneous \emph{pair creation}.
One can easily resolve the fermionic Klein paradox in the second quantized theory.
Splitting the Hilbert space into the particle and antiparticle subspaces and applying second quantization makes the quantum Hamiltonian positive definite.
The corresponding problem for bosons is much more serious.
If the classical Hamiltonian is not positive, it will not become positive by quantization.
Besides, in this case there is no positive scalar product preserved by the evolution, as is the case for Dirac fermions.
This typically leads to the so-called \emph{superradiance}.
In mathematical terms it means that the scattering operator has a norm greater than one, or it does not exist at all because the norm of the evolution grows all the time.

This paper is devoted to the study of the Klein--Gordon equation on rather general (possibly, non-stationary) spacetimes.
We construct both the classical propagators and certain families of non-classical propagators.
Let us first describe the basic steps of our construction of the classical propagators.
\begin{enumerate}[label=\arabic*.]
  \item
    We assume that there is a manifold $\Sigma$ such that the spacetime $M$ is diffeomorphic to $\RR \times \Sigma$.
    This diffeomorphism provides a global time function $t$ whose level sets $\Sigma_t$ are assumed to be spacelike.
    It also defines a flow whose generator $\partial_t$ is assumed to be timelike.
  \item
    We rewrite the Klein-Gordon equation as a (non-autonomous) 1st order equation for the Cauchy data on $\Sigma_t$.
    Thus the generator of the evolution can be written as a $2\times2$ matrix.
  \item
    We make various assumptions on the metric, electromagnetic and scalar potentials.
    The assumptions on their regularity are rather weak, however, they are global in spacetime.
    We assume that the positive mass condition holds for all times, that is, all instantaneous Hamiltonians have a strictly positive lower bound.
  \item
    We apply functional analytic methods from the theory of \emph{non-autonomous evolution equations}, as developed by Kato in~\cite{kato:hyperbolic}.
    Note that, unlike in~\cite{derezinski-siemssen}, in the non-stationary case we do not have a unique distinguished energy space.
    Instead, we have a whole time-dependent family of instantaneous energy Hilbert spaces describing the Cauchy data at each time.
    Under the assumptions we impose, these spaces can be identified with one another.
    They have a variable scalar product, but a common topology -- thus the Cauchy data at each time belong to a single \emph{Hilbertizable space}.
  \item
    The Pauli--Jordan propagator essentially coincides with one of the matrix elements of the evolution operator.
    One can then write down the forward and backward propagators by inserting the Heaviside function in the appropriate places.
    Thus if one uses the method of evolution equations, the Pauli--Jordan propagator becomes the central object, whereas in typical approaches found in the literature (\eg~\cite{bar}) the forward and backward propagators are obtained first and then used to define the Pauli--Jordan propagator.
    We find this (trivial) observation curious.
\end{enumerate}

The assumptions of our paper, in particular their global character and the positive mass assumptions, are adapted to the needs of non-classical propagators, which are our main interest.
However, if one is interested only in classical propagators, some of these assumptions can be be relaxed.

When the Hamiltonian is merely bounded from below, we can reduce the problem to the positive mass case by a perturbation argument.
Then one can construct the evolution, and hence also the classical propagators.
We remark about this fact at the end of Sect.~\ref{sec:evolution}.

Another point that can be relaxed are the global assumptions.
We know that the propagation of solutions to the Klein-Gordon equation has a finite speed -- this can be proven independently under weak assumption on the regularity, see e.g.\ Appx.~\ref{appx:finite-speed}.
Therefore, to construct the evolution, it is sufficient to have local information about our system.
We do not discuss this point further in our paper.

As already stated above, our main interest are the non-classical propagators.
Unfortunately, in the non-stationary case it is not obvious how to define them.
The most popular view on this subject says that instead of a single positive frequency bisolution one should consider a whole class of bisolutions locally similar to the Minkowski two-point function, known as \emph{Hadamard states}.
There exists a considerable literature about them; in particular we would like to mention~\cite{radzikowski,kay-wald}.
Properties of Hadamard states play a central role in most formulations of perturbation theory and renormalization on curved spacetimes, see \eg~\cite{hollands-wald1,hollands-wald2}.
Moreover, the expectation value of time-ordered fields in every Hadamard state is the integral kernel of an inverse of~$K$ and can be viewed as a possible generalization of the usual Feynman propagator to the generic case.

One of possibilities is to use spectral projections of the generator of the evolution at a fixed instance of time, as we describe in Sect.~\ref{sec:instantaneous}.
This allows us to define \emph{instantaneous positive and negative frequency bisolutions}, which yield the so-called \emph{instantaneous vacua}.
One also has the corresponding \emph{instantaneous Feynman inverses}.

One can criticize these propagators on physical grounds.
Not only do they depend on an arbitrary and unphysical choice of a preferred time, but it is a folklore knowledge that they are generally not Hadamard states.
In a forthcoming article~\cite{derezinski-siemssen:hadamard} we will show, using methods from our formalism, that an instantaneous positive frequency bisolution yields a Hadamard state if the Klein--Gordon operator~$\KG$ is infinitesimally stationary at the Cauchy surface where the positive/negative frequency splitting was performed.

Spacetimes that become asymptotically stationary in the past and the future form a class that in our opinion is especially natural from the point of view of quantum field theory and scattering theory.
For such spacetimes one can define positive/negative frequency bisolutions corresponding to the asymptotic past and future, see Sect.~\ref{sec:asymptotic}.
We can call them \emph{in-} and \emph{out-positive/negative frequency bisolutions}.
One can argue that the corresponding \emph{in-vacuum} yields the representation of incoming states (prepared in the experiment) and the corresponding \emph{out-vacuum} gives the representation of final observables (measured in the experiment).
Therefore, the in- and out states are not only distinguished, they also have a clear and important physical meaning.
If the spacetime becomes stationary sufficiently fast, it can be shown that the states thus defined are Hadamard~\cite{gerard-wrochna:inout}, see also~\cite{derezinski-siemssen:hadamard}.

As we described above, and is well-known, spacetimes with asymptotically stationary past and future posses two pairs of distinguished and physically well-motivated propagators: the in- and out- positive and negative frequency bisolution.
It is perhaps less known that a large class of such spacetimes possesses another pair of natural and physically motivated propagators: the so-called \emph{canonical Feynman} and \emph{anti-Feynman propagator (inverse)}.
The Feynman propagator appears naturally when we evaluate Feynman diagrams.
A study of these propagators will be presented in our following paper~\cite{derezinski-siemssen:feynman}, where the formalism and results of the present paper will play an important role.

Let us mention that the canonical Feynman and anti-Feynman propagator are   related to the intriguing and poorly understood question about the self-adjointness of the Klein--Gordon operator.
It is easy to see that the Klein-Gordon operator is Hermitian, however, the existence of a distinguished self-adjoint extension seems to be difficult to prove and is known only in special cases: in the static case \cite{derezinski-siemssen} and (since very recently) for a class of asymptotically Minkowskian spaces \cite{vasy:selfadjoint}.
Note that heuristically the canonical Feynman (resp. anti-Feynman) propagator is the boundary value at zero from above (resp. below) of the resolvent of the Klein-Gordon operator.
One could also argue that the adjective \emph{canonical} is not needed for both propagators, that they should simply be called \emph{the} Feynman and anti-Feynman propagator.

Let us compare our work with the literature.
The construction of classical propagators is described in numerous sources.
Typically, one shows first the well-posedness of the Cauchy problem.
Then the existence of the classical propagators and their properties easily follow, see \eg~\cite{kay,dimock-kay}, and also the more recent works~\cite{drago-gerard,gerard-wrochna:hadamard,gerard-oulghazi-wrochna}.
Standard methods include the Hadamard parametrix method~\cite{bar,friedlander} and energy estimates obtained via the divergence theorem.
Another popular method relies on the factorization of the Klein-Gordon operator into the product of 1st order scalar operators (see \eg\ the treatment of H\"ormander~\cite{hormander}, which also covers $n$th order hyperbolic equations).
A brief history of the Cauchy problem for hyperbolic equations with references to various approaches can be found in Notes to Chap.~XXIII of~\cite{hormander}.

In our opinion, the method of evolution equations used in this paper provides a natural and powerful approach to analyze the Klein--Gordon equation on curved spacetimes, especially concerning questions relevant to quantum field theory.
Therefore, we were greatly surprised that it is difficult to find a treatment of this problem similar to ours in the existing literatures.
We are only aware of one more publication where the methods of evolution equations have been applied to the problem at hand in the non-stationary case:
In~\cite{furlani}, Furlani constructs the evolution under quite restrictive assumptions, namely, assuming that Cauchy surfaces are compact and have a decreasing volume along a finite time-interval.
The treatment of some papers, such as by Dimock, Kay, and Gerard--Wrochna, may also resemble our method~\cite{dimock-kay,kay,gerard-oulghazi-wrochna,gerard-wrochna:hadamard}.
However, in almost all papers that we know, the existence of the evolution is taken for granted, is given by the local theory, and is not constructed within the formalism of evolution equations on some Banach spaces.

The literature devoted to classical propagators on curved spacetimes usually does not use a global functional analytic setting.
As we discussed above, from the point of view of classical propagators, the method of our paper seems to impose unnecessary limitations, because of the global assumptions on the spacetime.
However, to define and study non-classical propagators, some kind of global assumptions are usually indispensable.

Most authors do not consider low regularity situations (for an exception we refer to~\cite{sanchez-vickers})
For example, the propagators are typically understood from~$C^\infty_\comp(M)$ to~$C^\infty(M)$.
As far as we know, the constructions found in the literature require more stringent regularity assumptions than ours.

Throughout our paper we impose rather weak assumptions on the regularity of various objects (the metric, electromagnetic potential and the scalar potential).
Nevertheless, we did not write this work with any particular non-regular examples in mind, even though low regularity is present in some interesting physical applications (\eg, boundaries of astrophysical objects, shock waves) and singularities appear generically in solutions of the Einstein equation.
Instead, the main reason for the chosen approach is our conviction that weak assumptions play an important theoretical role, because they impose a certain discipline on a mathematical theory, forcing us to find better arguments and a more natural setting for the problem.

We think that our approach is rather natural and direct if one wants to treat the most simple examples of spacetimes (from the point of view of quantum field theory) such as local perturbations of Minkowski spacetime and cosmological spacetimes.
However, it is also flexible enough to treat some less obvious examples, such as certain non-globally hyperbolic spacetimes, including spacetimes with boundaries, provided we impose appropriate boundary conditions.
This includes for example compactifications of anti-de Sitter spacetime with appropriate conditions on its timelike boundary (\cf~\cite{dappiaggi-ferreira-marta,dappiaggi-drago-ferreira} for a recent discussion of boundary conditions on anti-de Sitter spacetime and spacetimes with a timelike boundary).

Finally, let us remark that Kato's theory of non-autonomous evolution equations has also been successfully applied in the context of quantum field theory for the Dirac equation on curved spacetimes, see \eg~\cite{haefner,nicolas}.
The Dirac equation is simpler in this respect than the Klein--Gordon equation.
For the Dirac equation there exists a natural Hilbert space.
For the (non-stationary) Klein--Gordon equation no such choice exists: one is forced to work with a family of Hilbertizable spaces.
Studying the evolution for the Klein--Gordon equation in time-dependent families of Hilbert spaces has also been fruitful in the context of spherical gravitational collapse (\ie, in static Schwarzschild spacetime with time-dependent boundary conditions), see \cite{bachelot} and references therein.

\subsection{Notation and conventions}

Throughout this paper we adopt essentially the same notations and conventions as in~\cite{derezinski-siemssen} but for the convenience of the reader we repeat the relevant conventions.
We also introduce some new notation.

Suppose that $T$ is an operator on a Banach space~$\mathcal{X}$.
We denote by $\Dom T$ its domain and by $\Ran T$ its range.
For its spectrum we write $\spec T$ and for the resolvent set $\reso T$.

Suppose that $T$ is a operator on a Hilbert space~$\mathcal{H}$ with inner product $(\,\cdot\;|\;\cdot\,)$.
If $T$ is positive, \ie, $(u \,|\, T u) \geq 0$, we write $T \geq 0$.
If also $\Ker T = \{0\}$, then we write $T > 0$.

A useful function is the so-called `Japanese bracket', defined as $\jnorm{T} \defn (1 + \abs{T}^2)^{1/2}$.

A topological vector space $\mathcal{X}$ is called \emph{Hilbertizable} if there exists a scalar product on $\mathcal{X}$ that determines its topology and makes it into a Hilbert space.
Clearly, two scalar products determine the topology of $\mathcal{X}$ iff they are equivalent.

The $p$-times continuously differentiable $\mathcal{X}$-valued functions on a manifold~$M$ are denoted $C^p(M; \mathcal{X})$; if $\mathcal{X} = \CC$, we simply write $C^p(M)$.
Sets of compactly supported or bounded functions are indicated by a subscript `c' or `b'.

$AC(\RR)$ denotes the set of absolutely continuous functions, \ie, functions whose distributional derivative belongs to $L_\loc^1(\RR)$.
$AC^1(\RR)$ denotes the set of functions whose distributional derivative belongs to $AC(\RR)$.

When calculating integrals, we denote by $\int'$ the `Cauchy principal value' at infinity, \eg,
\begin{equation*}
  \int'_{\im\RR} f(t)\, \dif t = \lim_{R \to \infty} \int_{-\im R}^{\im R} f(t)\, \dif t.
\end{equation*}
Observe that we pass to infinity symmetrically in the lower and upper integration limits.

Suppose we fix a positive density $\gamma$ on $M$.
The space $L^2(M,\gamma)$ of square-integrable functions on $M$ is then defined as the completion of~$C^\infty_\comp(M)$ with respect to the scalar product
\begin{equation*}
  (u \,|\, v)_\gamma \defn \int_M \conj{u}\, v\, \gamma,
  \quad
  u, v \in C^\infty_\comp(M).
\end{equation*}

If $g$ is the metric tensor $g$ on $M$ (of any signature), then we set $|g| \defn |\det [g_{\mu\nu}]|$.
$M$ is then equipped with a canonical density $|g|^{\frac12}$.
Sometimes it is however convenient to fix a density $\gamma$ independent of the metric tensor.

Often it is convenient to use the formalism of (complexified) half-densities on $M$.
If $\gamma$ is a positive density on $M$, then $\gamma^{\frac12}$ is a half-density.
The canonical example for a half-density on a pseudo-Riemannian manifold is $\abs{g}^\frac14$.
Since the integral over a density on a manifold is well-defined, half-densities come equipped with a natural $L^2$-structure
\begin{equation*}
  (\tilde{u} \,|\, \tilde{v}) = \int_M \conj{\tilde u}\, \tilde{v},
  \quad
  \tilde{u}, \tilde{v} \in C^\infty_\comp(\Omega^\frac12 M)
\end{equation*}
We denote by $L^2(\Omega^\frac12 M)$ the completion of $C^\infty_\comp(\Omega^\frac12 M)$ with respect to the corresponding norm.
Note that if we fix an everywhere positive density $\gamma$, then
\begin{equation*}
  L^2(M,\gamma) \ni u \mapsto \tilde{u} \defn u \gamma^{\frac12} \in L^2(\Omega^\frac12 M)
\end{equation*}
is the natural unitary identification of the $L^2$-space in the scalar formalism and in the half-density formalism.

The operator $D=-\im\partial$ acts naturally on scalars, and $D^\gamma=\gamma^{\frac12}D\gamma^{-\frac12}$ acts naturally on half-densities.

In our paper we generally prefer to use the half-density formalism rather than the scalar formalism.
The Klein--Gordon operator $\KG$ is presented in~\eqref{eq:klein-gordon} in the scalar formalism.
Transformed to the half-density formalism it is
\begin{equation}\label{eq:klein-gordon-half}
  \KG_{\frac12} \defn \abs{g}^\frac14 \KG \abs{g}^{-\frac14} = \abs{g}^{-\frac14} (D_\mu - A_\mu) \abs{g}^\frac12 g^{\mu\nu} (D_\nu - A_\nu)\abs{g}^{-\frac14} + Y.
\end{equation}
In what follows we drop the subscript $\frac12$ from $\KG_\frac12$ and by $\KG$ we will mean~\eqref{eq:klein-gordon-half}.

\section{Assumptions and setting}
\label{sec:assumptions}

\subsection{1+3 splitting}
\label{sub:1+3_splitting}

We consider smooth manifolds $M$ and $\Sigma$ such that there exists a (fixed) diffeomorphism $\RR \times \Sigma \to M$.
This means that we have a distinguished time function~$t$ on~$M$, and the leaves $\Sigma_t = \{t\} \times \Sigma$ provide a foliation of~$M$ with a family of diffeomorphisms $\epsilon_t : \Sigma \to \Sigma_t \subset M$.
We define the time vector field
\begin{equation*}
  \partial_t \defn \frac{\dif}{\dif t} \epsilon_t.
\end{equation*}
Note that $\dif t \cdot \partial_t=1$.

We assume that $M$ is equipped with a continuous Lorentzian metric $g$, \ie, $(M,g)$ is a \emph{spacetime}.
The restriction of~$g$ to the tangent space of~$\Sigma_t$ defines a time-dependent family of metrics on~$\Sigma$, denoted $g_\Sigma(t) \defn \epsilon_t^* g$.
We make the assumption that all $g_\Sigma(t)$ are Riemannian, or, equivalently, that the covector $\dif t$ is everywhere timelike.
This assumption allows us to define the \emph{lapse function}~$\alpha$:
\begin{equation*}
  \frac{1}{\alpha^2} \defn -g^{-1}(\dif t, \dif t) > 0.
\end{equation*}
Note that at this moment we do not assume that the vector $\partial_t$ is everywhere timelike, which is equivalent to
\begin{equation}\label{eq:timelike}
  g_\Sigma(\beta,\beta) < \alpha^2.
\end{equation}
This assumption will be forced on us later on by Assumption~\ref{asm:L-WW}.
The part of $\partial_t$ orthogonal to the leaves of the foliation is the \emph{shift vector}
\begin{equation*}
  \beta \defn \partial_t + \alpha^2 g^{-1}(\dif t, \,\cdot\,).
\end{equation*}

The inverse metric can now be written as
\begin{equation}\label{eq:metric}
  g^{-1} = -\frac{1}{\alpha^2} (\partial_t - \beta) \otimes (\partial_t - \beta) + g_\Sigma^{-1}.
\end{equation}
In coordinates, we have
\begin{align*}
  g_{\mu\nu} \dif x^\mu \dif x^\nu &= -\alpha^2 \dif t^2 + g_{\Sigma,ij} (\dif x^i + \beta^i \dif t) (\dif x^j + \beta^j \dif t), \\
  g^{\mu\nu} \partial_\mu \partial_\nu &= -\frac{1}{\alpha^2} (\partial_t - \beta^i \partial_i)^2 + g_\Sigma^{ij} \partial_i \partial_j.
\end{align*}

The generic notation for a point of~$M$ will be $(t,\vec x)$.
We often suppress the spatial dependence of objects defined on~$M$, \eg, we identify $f(t) = f(t, \cdot\,)$ for some function~$f$ on~$M$.
Sometimes we also suppress the time-dependence, but it should be kept in mind that the central quantities considered here, the metric~$g$, the electromagnetic potential~$A$ and the scalar potential~$Y$, generically are time-dependent.
Sometimes we denote derivatives with respect to~$t$ (\ie, the action of the vector field $\partial_t$) by a dot.

\subsection{Klein--Gordon operator}
\label{sub:KG_operator}

The main object of our paper is the Klein--Gordon operator~\eqref{eq:klein-gordon-half}.
Instead of the operator $\KG$ on $L^2(M)$, it is more convenient to work with the operator
\begin{equation*}
  \tilde\KG \defn \alpha \KG \alpha.
\end{equation*}

With the inverse metric expressed as~\eqref{eq:metric}, it can be written as
\begin{align*}
  \tilde\KG
  &= -\gamma^{-\frac12} (D_t - D_i \beta^i + V) \gamma (D_t - \beta^j D_j + V) \gamma^{-\frac12} \\&\quad + \gamma^{-\frac12} (D_i - A_i) \alpha^2 \gamma g_\Sigma^{ij} (D_j - A_j) \gamma^{-\frac12} + \alpha^2 Y \\
  &= -(D_t + W^*) (D_t + W) + L,
\end{align*}
where we introduced
\begin{align*}
  \gamma &\defn \alpha^{-2} \abs{g}^\frac12 = \alpha^{-1} \abs{g_\Sigma}^\frac12, \\
  V &\defn -A_0 + A_i \beta^i, \\
  W &\defn \beta^i D_i + V - \frac12 \gamma^{-1} (D_t \gamma - \beta^i D_i \gamma), \\
  L &\defn D_i^{A,\gamma\, *} \tilde{g}_\Sigma^{ij} D_j^{A,\gamma} + \tilde{Y}
\end{align*}
and we use the shorthands
\begin{align*}
  \tilde{g}_\Sigma^{ij}(t) &\defn \alpha(t)^2 g_\Sigma^{ij}(t), \\
  \tilde{Y}(t) &\defn \alpha(t)^2 Y(t), \\
  D^{A,\gamma}(t) &\defn \gamma(t)^\frac12 \bigl(D - A(t)\bigr) \gamma(t)^{-\frac12}.
\end{align*}
Clearly, propagators for $\tilde\KG$ induce corresponding propagators for $\KG$.

\subsection{First-order formalism}

For each $t \in \RR$, we (formally) define
\begin{equation*}
  B(t) \defn \begin{pmatrix} W(t) & \one \\ L(t) & W(t)^* \end{pmatrix}.
\end{equation*}
Setting $u_1(t) = u(t)$ and $u_2(t) = -(D_t + W(t)) u(t)$, we find that
\begin{equation*}
  \bigl( \partial_t + \im B(t) \bigr) \begin{pmatrix} u_1(t) \\ u_2(t) \end{pmatrix} = 0
\end{equation*}
if and only if $u$ is a (weak) solution of the Klein--Gordon equation $\tilde\KG u = 0$.
Therefore we occasionally call $\partial_t + \im B(t)$ the \emph{first-order Klein--Gordon operator}.
The half-densities $u_1(t)$ and $u_2(t)$ may be called the \emph{Cauchy data} for $u$ at time~$t$.

\subsection{Assumptions local in time}
\label{sub:asm-local}

\begin{assumption}\label{asm:L}
  We suppose that the following assumptions hold:
  \begin{enumerate}[label=1.\alph*.,ref=1.\alph*]
    \item For all $t \in \RR$, $L(t)$ extends to a positive invertible self-adjoint operator on $L^2(\Omega^\frac12\Sigma)$ (denoted by the same symbol). \label{asm:L-selfadj}
    \item There exists $a \in C(\RR)$ such that $a(t)<1$ and $\norm[\big]{W(t) L(t)^{-\frac12}} \leq a(t)$. \label{asm:L-WW}
    \item \label{asm:LW_growth_loc}
      There exists a positive $C \in L^1_\loc(\RR)$ such that for all $\abs{t-s} \leq 1$
      \begin{equation}\label{eq:LW_growth_loc}
        \norm[\big]{L(t)^{-\frac12} \bigl( L(t) - L(s) \bigr) L(t)^{-\frac12}} + 2 \norm[\big]{\bigl( W(t) - W(s) \bigr) L(t)^{-\frac12}} \leq \abs*{\int_s^t C(r)\, \dif r},
      \end{equation}
      where we place the absolute value on the right-hand side to account for the arbitrary order of~$t$ and~$s$.
    \item \label{asm:beta_cont}
      $t \mapsto \alpha(t)^{\pm1}$ are norm-continuous on $L(s)^{-\frac12} L^2(\Omega^\frac12\Sigma)$ for any $s \in \RR$, and $t \mapsto \dot\alpha(t)$ is norm-continuous on~$L^2(\Omega^\frac12\Sigma)$.
  \end{enumerate}
\end{assumption}

A few remarks about these assumptions are in order:

First, Assumption~\ref{asm:L-selfadj} can always be realized if $\gamma(t)^{-1}\partial_i\gamma(t)$, $A_i(t) \in L^2_\loc(\Sigma)$, $\tilde{g}_{ij}(t)\in L_\loc^\infty(\Sigma)$ and $\tilde{Y}(t) \in L^1_\loc(\Sigma)$ such that $\tilde{Y}(t)$ is bounded from below by a positive constant.
In that case $L(t)$ can be understood as the form
\begin{equation*}
  (u \,|\, L(t)\, v) = \int_\Sigma \Bigl( \bigl(\conj{D^{A,\gamma}_i(t)\, u}\bigr) \tilde{g}_\Sigma^{ij}(t) \bigl(D^{A,\gamma}_j(t)\, v\bigr) + \conj{u}\,\tilde{Y}(t)\, v \Bigr),
\end{equation*}
 on its (natural) maximal form domain $\Dom L(t)^\frac12 \supset C^\infty_\comp(\Omega^\frac12\Sigma)$ (but it is not generally clear if $C^\infty_\comp(\Omega^\frac12\Sigma)$ is a form core).
This form then defines a self-adjoint operator in the usual way.
The details of this construction are given in Appx.~\ref{appx:laplace}; its main aspects can be found in Thm.~VI.2.6 of~\cite{kato}.

Next, Assumption~\ref{asm:L-WW} means that $\norm{W(t)L(t)^{-\frac12}}<1$.
Thus the electrostatic potential $V(t)$ together with the variation of the metric expressed by $\gamma(t)^{-1} \dot\gamma(t)$ and the shift vector~$\beta$ cannot be too big compared to~$L(t)$.
This has to be true already on the level of the principal symbols of $W$ and $L$.
Therefore, for each $x = (t,\vec{x}) \in M$ and $p \in T_{\vec x}^*\Sigma_t$, we need to have
\begin{equation*}
  \abs[\big]{\beta^k(x) p_k \bigl(\tilde{g}_\Sigma^{ij}(x) p_i p_j\bigr)^{-\frac12}} < 1.
\end{equation*}
This is equivalent to
\begin{equation}\label{eq:timelike1}
  \tilde g_{\Sigma,ij} \beta^i \beta^j < 1,
\end{equation}
where $\tilde g_{\Sigma,ij} = \alpha^{-2} g_{\Sigma,ij}$ is the inverse of $\tilde g_\Sigma^{ij}$, and consequently \eqref{eq:timelike1} is equivalent to~\eqref{eq:timelike}.
Thus Assumption~\ref{asm:L-WW} implies that $\partial_t$ is timelike.
This excludes \eg\ the ergosphere region of Kerr spacetime in stationary coordinates -- in such a case one needs to switch to the non-stationary co-rotating coordinates.

Together, Assumptions \ref{asm:L-selfadj} and~\ref{asm:L-WW} guarantee that the Hamiltonian is positive and has a positive lower bound (the ``positive mass assumption'').
The positivity of the Hamiltonian and its positive lower bound has two aspects.
First, it is essentially necessary if we want to construct non-classical propagators.
Second, this assumption helps us to introduce a natural family of Hilbertizable spaces, which are used in the analysis of the evolution.
(A similar analysis would be possible with a positive Hamiltonian, but without a positive lower bound, however there would be some additional technical problems).

Nevertheless, as far as the derivation of the evolution and the classical propagators is concerned, Assumption~\ref{asm:L-selfadj} can be relaxed.
In fact, for the existence of the evolution it is sufficient that there exists a constant $b > 0$ such that these assumptions are satisfied by $L(t)+b$, see also Cor.~\ref{cor:bounded-below}.
In this case in general we do not have a positive Hamiltonian and our analysis of non-classical propagators does not apply.

Among other things, Assumption~\ref{asm:LW_growth_loc} guarantees that for any $t,s$ there exists $c(t,s)>0$ such that
\begin{equation}\label{eq:hilbert}
  L(t) \leq c(t,s) L(s).
\end{equation}
Therefore, for $\delta \in [-1,1]$ we can define the Hilbertizable spaces
\begin{equation*}
  \mathcal{K}^\delta \defn L(t)^{-\delta/2} L^2(\Omega^\frac12\Sigma),
\end{equation*}
where the Hilbertian structures on the right-hand side are equivalent for different $t$ because of~\eqref{eq:hilbert}.

Finally, Assumption~\ref{asm:beta_cont} implies the norm-continuity of $t \mapsto \alpha(t)^{\pm1}$ on $\mathcal{K}^\delta$ for $\delta \in [-1,1]$.
Indeed, by this assumption, $t \mapsto \alpha(t)^{\pm1}$ are norm-continuous on $\mathcal{K}^{\frac12}$, hence  by duality also on $\mathcal{K}^{-\frac12}$, and then we can interpolate using, \eg, the Heinz--Kato inequality (Thm.~\ref{thm:heinz-kato}).

While it should be obvious how Assumption~\ref{asm:L-selfadj}, \ref{asm:L-WW} and~\ref{asm:beta_cont} can be realized in an example, Assumption~\ref{asm:LW_growth_loc} is slightly less obvious.
Therefore in Appx.~\ref{appx:assumptions} we briefly explain how Assumption~\ref{asm:LW_growth_loc} can follow from more concrete assumptions on the metric, the vector potential and the scalar potential.

\subsection{Assumptions global in time}
\label{sub:asm-global}

While we always require that Assumption \ref{asm:L} holds, the following additional assumptions are only imposed when we derive asymptotic properties of propagators.

\begin{assumption}\label{asm:asymptotic}
  \mbox{}
  \begin{enumerate}[label=2.\alph*.,ref=2.\alph*]
    \item $L(t)$ is uniformly bounded away from zero.
    \item There exists $a < 1$ such that $\norm[\big]{W(t) L(t)^{-\frac12}} \leq a$ for all $t$.
    \item \label{asm:LW_growth}
      There exists a positive $C \in L^1(\RR)$ such that for all $t,s \in \RR$
      \begin{equation*}
        \norm[\big]{L(t)^{-\frac12} \bigl( L(t) - L(s) \bigr) L(t)^{-\frac12}} + 2 \norm[\big]{\bigl( W(t) - W(s) \bigr) L(t)^{-\frac12}} \leq \abs*{\int_s^t C(r)\, \dif r},
      \end{equation*}
      where we place the absolute value on the right-hand side to account for the arbitrary order of~$t$ and~$s$.
    \item \label{asm:beta_bnd}
      $t \mapsto \alpha(t)^{\pm1}$ are uniformly bounded on~$\mathcal{K}^1$ and $t \mapsto \dot\alpha$ is uniformly bounded on~$\mathcal{K}^0$.
  \end{enumerate}
\end{assumption}

Note that, by the same argument as for Assumption~\ref{asm:beta_cont}, one can show that Assumption~\ref{asm:beta_bnd} implies the uniform boundedness of $t \mapsto \alpha(t)^{\pm1}$ on $\mathcal{K}^\delta$ for $\delta \in [-1,1]$.

\section{The energy space and the dynamical space}
\label{sec:spaces}

We will occasionally use the Hilbert space
\begin{equation*}
  \mathcal{H} \defn L^2(\Omega^\frac12\Sigma) \oplus L^2(\Omega^\frac12\Sigma) = \mathcal{K}^0 \oplus \mathcal{K}^0
\end{equation*}
with the canonical inner product also denoted by $(\,\cdot\;|\;\cdot\,)$ and the corresponding norm $\norm{\,\cdot\,}$.

The Hilbert space $\mathcal{H}$ plays only an auxiliary role in our work.
More important are the Hilbertizable spaces $\mathcal{H}_\lambda$, $\lambda \in [-1,1]$, defined as
\begin{equation}\label{eq:H_alpha}
  \mathcal{H}_\lambda \defn \mathcal{K}^{(\lambda+1)/2} \oplus \mathcal{K}^{(\lambda-1)/2}.
\end{equation}
Note that for any $t$
\begin{equation}\label{eq:scale}
  \mathcal{H}_\lambda = \bigl( L(t) \oplus L(t) \bigr)^{-\lambda/4} \mathcal{H}_0,
  \quad
  \lambda \in [-1,1].
\end{equation}
We will treat the space $\mathcal{H}_0$ as the central element of the family~\eqref{eq:scale}, identifying $\mathcal{H}_0$ with $\mathcal{H}^{\mathrlap{*}}{}_0$, the \emph{antidual} of $\mathcal{H}_0$ (the space of bounded antilinear functionals on $\mathcal{H}_0$).
Then we have a natural identification of $\mathcal{H}_{-\lambda}$ with $\mathcal{H}^{\mathrlap{*}}{}_{\lambda}$.

The central role in this work is played by the \emph{energy space}, the \emph{dynamical space} and the antidual of the energy space:
\begin{subequations}\label{eq:free-spaces}\begin{align}
  \Hen &\defn \mathcal{H}_1 = \bigl( L(t)^{-\frac12} \oplus \one \bigr) \mathcal{H} = H_0(t)^{-\frac12} \mathcal{H}, \\
  \Hdyn &\defn \mathcal{H}_0 = \bigl( L(t)^{-\frac14} \oplus L(t)^{\frac14} \bigr) \mathcal{H}, \\
  \Hne &\defn \mathcal{H}_{-1} = \bigl( \one \oplus L(t)^\frac12 \bigr) \mathcal{H} = \bigl( Q H_0(t) Q \bigr)^\frac12 \mathcal{H},
\end{align}\end{subequations}
where we set
\begin{equation*}
  H_0(t) \defn L(t) \oplus \one = \twomat{L(t)}{0}{0}{\one},
\end{equation*}
and we also used the \emph{charge form}
\begin{equation*}
  (u \,|\, Q v) \defn (u_1 \,|\, v_2) + (u_2 \,|\, v_1),
  \quad
  Q \defn \begin{pmatrix} 0 & \one \\ \one & 0 \end{pmatrix}.
\end{equation*}
It is evident that the charge form is bounded on~$\mathcal{H}$.
More importantly, it is also bounded on~$\Hdyn$ (but, \eg, not on~$\Hen$).

Note that
\begin{equation*}
  \Im\, (u \,|\, Q v) = \frac1{2\im} \bigl( (u \,|\, Q v) - (v \,|\, Q u) \bigr)
\end{equation*}
is a symplectic form on~$\Hdyn$.
Therefore, the formalism based on the charge form is equivalent to the symplectic formalism, commonly used in the literature.

\section{Instantaneous energy spaces and instantaneous dynamical spaces}
\label{sec:instantaneous_spaces}

An important role in our paper is played by the instantaneous Hamiltonian, defined formally for each $t$ as
\begin{equation*}
  H(t) = Q B(t) = B(t)^* Q.
\end{equation*}
One can rigorously define $H(t)$ as a form bounded perturbation of $H_0(t)$:
\begin{proposition}
  The operator
  \begin{equation*}
    H(t) \defn \begin{pmatrix} L(t) & W(t)^* \\ W(t) & \one \end{pmatrix}
  \end{equation*}
  is self-adjoint on $\mathcal{H}$ with the form domain $\Hen$.
  We have
  \begin{equation}\label{eq:H-H0-ineq}
    \bigl( 1-a(t) \bigr) H_0(t) \leq H(t) \leq \bigl( 1+a(t) \bigr) H_0(t),
  \end{equation}
  where $0 \leq a(t) < 1$ was introduced in Assumption~\ref{asm:L-WW}.
\end{proposition}
\begin{proof}
  We show only the right-hand side of the inequality~\eqref{eq:H-H0-ineq}.
  Set $u = \left(\!\begin{smallmatrix} u_1 \\ u_2 \end{smallmatrix}\!\right)$.
  Using the Cauchy--Schwarz inequality and Assumption~\ref{asm:L-WW}, we find
  \begin{align*}
    (u \,|\, H(t)\, u)
    &\leq \norm{L(t)^\frac12 u_1}^2 + \norm{u_2}^2 + 2 \norm{W(t)\, u_1}\norm{u_2} \\
    &\leq \norm{L(t)^\frac12 u_1}^2 + \norm{u_2}^2 + 2 a(t) \norm{L(t)^\frac12 u_1} \norm{u_2} \\
    &\leq \bigl(1+a(t)\bigr) \bigl( \norm{L(t)^\frac12 u_1}^2 + \norm{u_2}^2 \bigr) \\
    &= \bigl(1+a(t)\bigr) (u \,|\, H_0(t)\, u).
  \end{align*}
\end{proof}

We define for each time $t \in \RR$ the \emph{(instantaneous) energy scalar products} given by
\begin{equation*}
  (u \,|\, v)_{\en,t} \defn (u \,|\, H(t) v)
 \end{equation*}
on~$\Hen$.
By~\eqref{eq:H-H0-ineq} the scalar product $(\,\cdot\;|\;\cdot\,)_{\en,t}$ is compatible with the topology of $\Hen$.
We call the resulting Hilbert space the \emph{instantaneous energy space at $t$} and denote it by $\Hen[t]$.

Similarly, we can also define the operator $Q H(t)^{-1} Q$.
We find that its form domain is~$\Hne$.
Indeed,
\begin{equation}\label{eq:QH-QH0-ineq}
  \bigl( 1+a(t) \bigr)^{-1} Q H_0(t)^{-1} Q \leq Q H(t)^{-1} Q \leq \bigl( 1-a(t) \bigr)^{-1} Q H_0(t)^{-1} Q.
\end{equation}
Then we define for each $t$ the scalar product
\begin{equation*}
  (u \,|\, v)_{\en^*,t} \defn (u \,|\,Q H(t)^{-1} Q v)
\end{equation*}
and note that it is compatible with the topology of $\Hne$; we denote the resulting Hilbert space by~$\Hne[t]$.

The central operator in this work is $B(t)$.
In the next section we construct the evolution generated by~$B(t)$, solving the first-order Klein--Gordon equation.

\begin{proposition}\label{prop:B}
  Considered as an operator on~$\Hne[t]$ with domain~$\Hen$,
  \begin{equation*}
    B(t) \defn \twomat{W(t)}{\one}{L(t)}{W(t)^*}
  \end{equation*}
  is self-adjoint and $0$ is in its resolvent set.
\end{proposition}
\begin{proof}
  For notational simplicity, we drop the time-dependence of $B(t)$ and the other objects.

  First note that, by definition, $H_0(t)^{-\frac12} = (L(t)^{-\frac12} \oplus \one)$ maps $\mathcal{H}$ to $\Hen$, and $(Q H_0(t) Q)^{-\frac12} = (\one \oplus L(t)^{-\frac12})$ maps $\Hne$ to $\mathcal{H}$.
  Now, to check that $B(t)$ is well-defined, we calculate
  \begin{equation*}
    \twomat{\one}{0}{0}{L^{-\frac12}}
    \twomat{W}{\one}{L}{W^*}
    \twomat{L^{-\frac12}}{0}{0}{\one}
    =
    \twomat{W L^{-\frac12}}{\one}{\one}{L^{-\frac12} W^*},
  \end{equation*}
  which is bounded by Assumption~\ref{asm:L-WW}.

  Next, we show that $0 \in \reso B$, and consequently also that $B$ is closed.
  We rewrite $B$ as
  \begin{equation*}
    B
    =
    \twomat{\one}{0}{W^*}{\one}
    \twomat{0}{\one}{L-W^*W}{0}
    \twomat{\one}{0}{W}{\one}
  \end{equation*}
  and check that $B^{-1}$ is bounded from $\Hne$ to $\Hen$:
  \begin{equation*}
    \twomat{L^\frac12}{0}{0}{\one}
    B^{-1}
    \twomat{\one}{0}{0}{L^\frac12}
    =
    \twomat{\one}{0}{-L^{-\frac12} W^*}{\one}
    \twomat{0}{\bigl( \one - L^{-\frac12} W^*W L^{-\frac12} \bigr)^{-1}}{\one}{0}
    \twomat{\one}{0}{-W L^{-\frac12}}{\one},
  \end{equation*}
  where the first and last factor on the right-hand side are bounded by Assumption~\ref{asm:L-WW}, and
  \begin{equation*}
    \one - L^{-\frac12} W^*W L^{-\frac12}
  \end{equation*}
  is invertible because $\norm{L^{-\frac12} W^*W L^{-\frac12}} < 1$, also by Assumption~\ref{asm:L-WW}.

  Finally, we check that $B$ is Hermitian on $\Hne$.
  We calculate
  \begin{equation*}
    (Q H Q)^{-1} B^{-1}
    = (B Q H Q)^{-1}
    = (Q H Q H Q)^{-1}
    = (Q H Q B^*)^{-1}
    = B^{*-1} (Q H Q)^{-1}.
  \end{equation*}
\end{proof}

We can now define for each time $t \in \RR$ a whole scale of Hilbert spaces
\begin{equation*}
  \mathcal{H}_{\lambda,t} \defn \abs{B(t)}^{-(1+\lambda)/2} \Hne[t],
  \quad
  \lambda \in \RR,
\end{equation*}
with scalar products
\begin{equation*}
  (u\,|\,v)_{\lambda,t} \defn \bigl(u\,\big|\,\abs{B(t)}^{1+\lambda} v\bigr)_{\en^*,t},
  \quad
  u, v \in \mathcal{H}_{\lambda,t}.
\end{equation*}
Above we performed the polar decomposition with respect to the Hilbert space~$\Hne[t]$, where we have
\begin{equation*}
  \abs{B(t)} = \sqrt{B(t)^2} = \sqrt{Q H(t) Q H(t)}.
\end{equation*}

It follows from its definition, that $B(t)$ extends/restricts to a self-adjoint operator on each of the spaces~$\mathcal{H}_{\lambda,t}$.
When $B(t)$ is interpreted as an operator on~$\mathcal{H}_{\lambda,t}$, its domain is $\mathcal{H}_{\lambda+2,t}$.

Clearly the scales $\mathcal{H}_{\lambda,t}$ contain $\Hne[t] = \mathcal{H}_{-1,t}$.
They also contain the (instantaneous) energy spaces $\Hen[t] = \mathcal{H}_{1,t}$, because a short calculation shows $H(t) = Q H(t)^{-1} Q \abs{B(t)}^2$.
Furthermore, we define the \emph{(instantaneous) dynamical spaces}
\begin{equation*}
  \Hdyn[t] \defn \mathcal{H}_{0,t},
\end{equation*}
which are treated as the central spaces in these scales.
Note that $\Hdyn[t]$ is the form domain of~$B(t)$.
We identify $\mathcal{H}^{\mathrlap{*}}{}_{0,t}$ with $\mathcal{H}_{0,t}$, and hence $\mathcal{H}^{\mathrlap{*}}{}_{\lambda,t}$ is identified with $\mathcal{H}_{-\lambda,t}$.
Thus we obtain the rigged Hilbert space setting
\begin{equation*}
  \Hen[t] \subset \Hdyn[t] \subset \Hne[t].
\end{equation*}

\begin{proposition}
  In the sense of Hilbertizable spaces, we have
  \begin{equation}\label{eq:interpo}
    \mathcal{H}_{\lambda,t} = \mathcal{H}_\lambda,
    \quad
    \lambda \in [-1,1],
  \end{equation}
  thus justifying our notation.
  In particular,
  \begin{equation*}
    \Hen[t] = \Hen,
    \quad
    \Hdyn[t] = \Hdyn,
    \quad
    \Hne[t] = \Hne.
  \end{equation*}
\end{proposition}
\begin{proof}
  It follows from~\eqref{eq:H-H0-ineq} and~\eqref{eq:QH-QH0-ineq} that $\Hen[t] = \Hen$ and $\Hne[t] = \Hne$.
  Since both $L(t)^\frac12 \oplus L(t)^\frac12$ and~$\abs{B}$ can be understood as invertible bounded operators from~$\Hen$ to~$\Hne$, there exists $c > 1$ such that
  \begin{equation*}
    c^{-1} \norm[\big]{\bigl( L(t) \oplus L(t) \bigr)^\frac12 u}_{\en^*} \leq \norm[\big]{\abs{B(t)} u}_{\en^*} \leq c \norm[\big]{\bigl( L(t) \oplus L(t) \bigr)^\frac12 u}_{\en^*}
  \end{equation*}
  By interpolation (\eg, using the Heinz--Kato inequality, Thm.~\ref{thm:heinz-kato}),
  \begin{equation*}
    c^{-\delta} \norm[\big]{\bigl( L(t) \oplus L(t) \bigr)^{\delta/2} u}_{\en^*} \leq \norm[\big]{\abs{B(t)}^\delta u}_{\en^*} \leq c^\delta \norm[\big]{\bigl( L(t) \oplus L(t) \bigr)^{\delta/2} u}_{\en^*}
  \end{equation*}
  for $\delta \in [0,1]$.
  It follows that the norms for $\mathcal{H}_\lambda$ and $\mathcal{H}_{\lambda,t}$ with $\lambda \in [-1,1]$ are equivalent and thus~\eqref{eq:interpo} follows.
\end{proof}

Note that for $\abs{\lambda}>1$ the spaces $\mathcal{H}_{\lambda,t}$ may depend on $t$ and do not have to coincide with $\mathcal{H}_\lambda$.

\section{Evolution}
\label{sec:evolution}

In the last section we laid the foundations for an application of the theory of non-autonomous evolution equations to the situation at hand, \ie, the first-order Klein--Gordon equation
\begin{equation*}
  \partial_t u(t) + \im B(t) u(t) = 0.
\end{equation*}
Autonomous evolution equations (\viz, with a time-independent generator) posses a well\hyp{}understood theory in terms of the theory of strongly continuous semigroups and groups.
The theory for non-autonomous evolution equations is significantly more complicated and subtle.
In Appx.~\ref{appx:evolution} we discuss the relevant results based on the work of Kato~\cite{kato:hyperbolic}.

Here we apply Thm.~\ref{thm:evolution_selfadjoint} to the operator~$B(t)$ on the spaces
\begin{equation}\label{eq:spaces-evolution}
  \mathcal{X}_t = \Hne[t]
  \quad\text{and}\quad
  \mathcal{Y}_t = \Hen[t].
\end{equation}
For this purpose, we need to check whether the conditions~\ref{item:evolution-selfadj:a}--\ref{item:evolution-selfadj:c} of Thm.~\ref{thm:evolution_selfadjoint} hold.
The self-adjointness condition~\ref{item:evolution-selfadj:c} is clearly true, see Sect.~\ref{sec:spaces}.
The next proposition implies that condition~\ref{item:evolution-selfadj:b}, a continuity condition on the norms of the Hilbert spaces~$\Hen[t]$ and~$\Hne[t]$, holds:

\begin{proposition}\label{prop:norm_growth}
  Let $C \in L^1_\loc(\RR)$ as in Assumption~\ref{asm:LW_growth_loc},
  $a(t) \in C(\RR)$ as in Assumption~\ref{asm:L-WW} and $\abs{t-s} \leq 1$ with $t \geq s$.
  Set
  \begin{equation*}
    c_{s,t} \defn \sup_{\tau \in [s,t]} \bigl(1-a(\tau)\bigr)^{-1}.
  \end{equation*}
  Then, for $\lambda \in [-1,1]$,
  \begin{equation}\label{eq:norm_growth_in_time}
    \norm{u}_{\lambda,t}\, \exp\left( -c_{s,t} \!\int_s^t\! C(\tau)\, \dif\tau \right) \leq \norm{u}_{\lambda,s} \leq \norm{u}_{\lambda,t}\, \exp\left( c_{s,t} \!\int_s^t\! C(\tau)\, \dif\tau \right).
  \end{equation}
\end{proposition}
\begin{proof}
  First we show~\eqref{eq:norm_growth_in_time} for $\lambda = 1$, \ie, for the energy space.

  By Assumption~\ref{asm:LW_growth_loc}, we have
  \begin{align}
    \MoveEqLeft\norm[\big]{\bigl( L(t)^{-\frac12} \oplus \one \bigr) \bigl( H(s) - H(t) \bigr) \bigl( L(t)^{-\frac12} \oplus \one \bigr)} \nonumber \\
    &\leq \norm[\big]{L(t)^{-\frac12} \bigl( L(s) - L(t) \bigr) L(t)^{-\frac12}} + 2\norm[\big]{\bigl( W(s) - W(t) \bigr) L(t)^{-\frac12}} \nonumber \\
    &\leq \int_s^t\! C(\tau)\, \dif\tau.\label{eq:growth_estimate1}
  \end{align}
  Eq.~\eqref{eq:H-H0-ineq} then implies that
  \begin{equation}\label{eq:growth_estimate2}
    \norm[\big]{H(t)^{-\frac12}\bigl(L(t)\oplus\one\bigr)H(t)^{-\frac12}} \leq c_{s,t}.
  \end{equation}
  Putting together~\eqref{eq:growth_estimate1} and~\eqref{eq:growth_estimate2}, we obtain
  \begin{equation*}
    \norm[\big]{H(t)^{-\frac12} \bigl( H(s) - H(t) \bigr) H(t)^{-\frac12}}
    \leq c_{s,t} \!\int_s^t\! C(\tau)\, \dif\tau.
  \end{equation*}
  Consequently we have
  \begin{equation*}
    \abs[\big]{\norm{u}_{\en,s}^2 - \norm{u}_{\en,t}^2} \leq \norm{u}_{\en,t}^2\, \left( c_{s,t} \!\int_s^t\! C(\tau)\, \dif\tau \right).
  \end{equation*}
  Therefore
  \begin{align*}
    \norm{u}_{\en,s}^2
    &\leq \norm{u}_{\en,t}^2\, \left( 1 + c_{s,t} \!\int_s^t\! C(\tau)\, \dif\tau \right) \\
    &\leq \norm{u}_{\en,t}^2\, \exp\left( c_{s,t} \!\int_s^t\! C(\tau)\, \dif\tau \right)
  \intertext{and, exchanging the role of $t$ and $s$, we can similarly derive}
    \norm{u}_{\en,s}^2
    &\geq \norm{u}_{\en,t}^2\, \exp\left( -c_{s,t} \!\int_s^t\! C(\tau)\, \dif\tau \right),
  \end{align*}
  so that the inequality~\eqref{eq:norm_growth_in_time} for $\lambda=1$ follows.

  For $\lambda = -1$ the inequality follows by duality.
  Using interpolation, we can then extend the inequality to the remaining values of $\lambda$.
\end{proof}

To show that the condition~\ref{item:evolution-selfadj:a} of Thm.~\ref{thm:evolution_selfadjoint} holds, we only need to show the norm-continuity of $t \mapsto B(t)$; the remaining statements are obvious.

\begin{proposition}\label{prop:norm_cont}
  With $C \in L^1_\loc(\RR)$ as in Assumption~\ref{asm:LW_growth_loc}, $c_{s,t}$ as in~\eqref{eq:H-H0-ineq} and $\abs{t-s} \leq 1$
  \begin{equation*}
    \norm[\big]{\bigl( B(s) - B(t) \bigr) u}_{\en^*,t} \leq \norm{u}_{\en,t}\, \abs*{c_{s,t} \!\int_s^t\! C(\tau)\, \dif\tau},
  \end{equation*}
  where we place the absolute value on the right-hand side because $t \geq s$ or $t \leq s$.
  In particular, $t \mapsto B(t)$ is norm-continuous as an operator from $\Hen[t]$ to $\Hne[t]$.
\end{proposition}
\begin{proof}
  We reduce the problem to the inequalities
  \begin{align*}
    \MoveEqLeft \norm[\big]{\bigl( \one \oplus L(t)^{-\frac12} \bigr) \bigl( B(s) - B(t) \bigr) \bigl( L(t)^{-\frac12} \oplus \one \bigr) } \\
    &= \norm[\big]{Q \bigl( L(t)^{-\frac12} \oplus \one \bigr) Q \bigl( B(s) - B(t) \bigr) \bigl( L(t)^{-\frac12} \oplus \one \bigr) } \\
    &\leq \norm[\big]{\bigl( L(t)^{-\frac12} \oplus \one \bigr) \bigl( H(s) - H(t) \bigr) \bigl( L(t)^{-\frac12} \oplus \one \bigr) } \\
    &\leq \abs*{\int_s^t\! C(\tau)\, \dif\tau}
  \end{align*}
  and proceed similar as in the proof of Prop.~\ref{prop:norm_growth}.
  Since the integral is continuous, the required norm-continuity follows.
\end{proof}

It follows that we can globally define an evolution for $B(t)$:
\begin{theorem}\label{thm:evolution}
  There exists a unique, strongly continuous family of bounded operators $\{U(t,s)\}_{s,t \in \RR}$ on $\Hne$, the \emph{evolution generated by} $B(t)$, with the following properties:
  \begin{enumerate}
    \item
      For all $r,s,t \in \RR$, we have the identities
      \begin{equation}\label{eq:U_composition}
        U(t,t) = \one,
        \quad
        U(t,r) U(r,s) = U(t,s).
      \end{equation}
    \item
      For $\lambda \in [-1,1]$, $U(t,s) \mathcal{H}_\lambda \subset \mathcal{H}_\lambda$, $(t,s) \mapsto U(t,s)$ is strongly $\mathcal{H}_\lambda$-continuous and satisfies the bound
      \begin{subequations}\label{eq:U_bound}\begin{align}
        \norm{U(t,r)}_{\lambda,s} &\leq \exp\left( 2c_{r,t} \!\int_r^t\! C(\tau)\, \dif\tau \right), \\
        \norm{U(r,t)}_{\lambda,s} &\leq \exp\left( 2c_{r,t} \!\int_r^t\! C(\tau)\, \dif\tau \right)
      \end{align}\end{subequations}
      with $C, c$ as in Prop.~\ref{prop:norm_growth} and $r \leq s \leq t$ where $\abs{t-r} \leq 1$.
    \item \label{item:evolution:3}
      For all $u \in \Hen$, $U(t,s) u$ is continuously differentiable in $s,t \in \RR$ with respect to the strong topology of~$\Hne$ and it satisfies
      \begin{subequations}\label{eq:U_sol}\begin{align}
         \im \partial_t U(t,s) u &= B(t) U(t,s) u, \\
        -\im \partial_s U(t,s) u &= U(t,s) B(s) u.
      \end{align}\end{subequations}
  \end{enumerate}
\end{theorem}
\begin{proof}
  Props.~\ref{prop:norm_growth} and~\ref{prop:norm_cont} as well as the results of Sect.~\ref{sec:spaces} show that Thm.~\ref{thm:evolution_selfadjoint} can be applied to our operator $B(t)$ understood as an operator from $\Hen$ to $\Hne$ (or, equivalently, as a form on $\Hdyn$ with form domain $\Hen$).
  We thus obtain for every sufficiently small compact interval $I \subset \RR$ an evolution $U(t,s)$ with the properties~\ref{item:evolution-selfadj:1}--\ref{item:evolution-selfadj:4} of Thm.~\ref{thm:evolution_selfadjoint}.
  In particular, we have for $r,t \in I$ and $r \leq s \leq t$
  \begin{align*}
    \norm{U(t,r)}_{\en,s} &\leq \exp\left( 2c_{r,t} \!\int_r^t\! C(\tau)\, \dif\tau \right), \\
    \norm{U(t,r)}_{\en^*,s} &\leq \exp\left( 2c_{r,t} \!\int_r^t\! C(\tau)\, \dif\tau \right).
  \end{align*}
  The same bounds also hold for $\norm{U(r,t)}_{\en,s}$ and $\norm{U(r,t)}_{\en^*,s}$.
  By interpolation we find~\eqref{eq:U_bound}.

  We cover $\RR$ by compact intervals.
  Using the identity~\eqref{eq:U_composition}, we thereby define the evolution $U(t,s)$ on the whole real axis by gluing.
  For finite $s,t$, it has the properties~\ref{item:evolution-selfadj:1}--\ref{item:evolution-selfadj:4} of Thm.~\ref{thm:evolution_selfadjoint}.
\end{proof}

Eq.~\eqref{eq:U_bound} states that $U(t,s)$ is bounded for finite $t,s$.
To obtain stronger results later, we can choose more stringent assumptions:
\begin{corollary}\label{cor:evolution-bound}
  If Assumption~\ref{asm:LW_growth} holds, and we set $C_1(t) \defn 2 (1-a)^{-1} C(t)$, then
  \begin{equation*}
    \norm{U(t,s)}_{\lambda,r} \leq \exp\left( \int_\RR\! C_1(\tau)\, \dif\tau \right)
  \end{equation*}
  for all $r,s,t \in \RR$ and any $\lambda \in [-1,1]$.
\end{corollary}

In Assumption~\ref{asm:L-selfadj} we supposed that $L(t)$ is positive and invertible.
Actually, the main results of this section remain true if $L(t)$ is only bounded from below:

\begin{corollary}\label{cor:bounded-below}
  Instead of Assumptions~\ref{asm:L}, suppose that there exists a constant $b > 0$ such that Assumptions~\ref{asm:L} hold for $L(t)+b$.
  Then Thm.~\ref{thm:evolution} holds with respect to the scale of Hilbert spaces and constants obtained from $L(t)+b$, and with the bounds~\eqref{eq:U_bound} replaced by
  \begin{align*}
    \norm{U(t,r)}_{\lambda,s} &\leq \exp\left( 2c_{r,t} \!\int_r^t\! C(\tau)\, \dif\tau + (t-r)b\norm[\big]{\bigl(L(s)+b\bigr)^{-\frac12}} \right), \\
    \norm{U(r,t)}_{\lambda,s} &\leq \exp\left( 2c_{r,t} \!\int_r^t\! C(\tau)\, \dif\tau + (t-r)b\norm[\big]{\bigl(L(s)+b\big)^{-\frac12}} \right).
  \end{align*}
\end{corollary}
\begin{proof}
  Replacing $L(t)$ in $B(t)$ by $L(t)+b$, we obtain a new operator
  \begin{equation*}
    B_b(t) = B(t) + \begin{pmatrix}
      0 & 0 \\ b & 0
    \end{pmatrix}.
  \end{equation*}
  We also replace $L(t)$ by $L(t)+b$ in all definitions concerning the Hilbert and Hilbertizable spaces that we need.
  According to Thm.~\ref{thm:evolution}, $B_b(t)$ has an evolution $U_b(t,s)$ with the properties stated in the theorem.
  Note that $\bigl(\begin{smallmatrix} 0 & 0 \\ b & 0 \end{smallmatrix}\bigr)$ is a bounded operator on $\mathcal{H}_\en^*$.
  Indeed, this follows from the boundedness of
  \begin{equation*}
    \begin{pmatrix}
      \one & 0 \\ 0 & L^{-\frac12}
    \end{pmatrix} \begin{pmatrix}
      0 & 0 \\ b & 0
    \end{pmatrix} \begin{pmatrix}
      \one & 0 \\ 0 & L^{\frac12}
    \end{pmatrix}= \begin{pmatrix}
      0 & 0 \\ L^{-\frac12}b & 0
    \end{pmatrix}
  \end{equation*}
  on $\mathcal{H}$.
  Since $B(t)$ is a bounded perturbation of~$B_b(t)$, we can apply Thm.~\ref{thm:perturbed_evolution} to find the evolution for~$B(t)$.
\end{proof}

\begin{remark}
  Our choice of spaces~\eqref{eq:spaces-evolution} to prove Thm.~\ref{thm:evolution} is natural, especially given our low regularity setup.
  Under more restrictive assumptions on the smoothness and boundedness of coefficients of the Klein--Gordon operator~$\KG$, other spaces in the scale $\mathcal{H}_{\lambda,t}$, $\lambda \in \RR$, could be used.
  This would lead to improved regularity results of the type $U(t,s) \mathcal{H}_\lambda \subset \mathcal{H}_\lambda$ and continuous differentiability of $U(t,s) \mathcal{H}_\lambda$ in~$\mathcal{H}_{\lambda-2}$.
\end{remark}

\begin{remark}
  In the stationary case, \ie, if $B$ does not depend on time, there exist distinguished Hilbert spaces $\mathcal{H}_\lambda$ and the evolution family $U(t,s)$ simplifies to a unitary group $U(t-s) = U(t,s)$ on $\mathcal{H}_\lambda$.
\end{remark}

\section{Solutions of the Klein--Gordon equation}

Solutions of the Klein--Gordon equation are closely related to solutions of the first-order Klein--Gordon equation.

Let us introduce the projection onto the second component:
\begin{equation*}
  \pi_2 \begin{pmatrix} u_1 \\ u_2 \end{pmatrix} \defn u_2,
\end{equation*}
We also define embeddings
\begin{equation*}
  \iota_2 u \defn \begin{pmatrix} 0 \\ u \end{pmatrix},
  \quad
  \rho u \defn \begin{pmatrix} u \\ -(D_t + W) u \end{pmatrix}.
\end{equation*}
A formal calculation then shows that\footnote{Note that there is a sign error in the corresponding equation in~\cite{derezinski-siemssen} which also affects the definition of the associated propagators.}
\begin{equation}\label{eq:KG1-KG}
  \tilde\KG = -\im \pi_2 (\partial_t + \im B) \rho
  \quad\text{and}\quad
  \KG = -\im \alpha^{-1} \pi_2 (\partial_t + \im B) \rho \alpha^{-1}.
\end{equation}
Therefore, if $\KG u = f$ or, equivalently, $\tilde\KG \tilde{u} = \tilde{f}$ with $\tilde{u} = \alpha^{-1} u$, $\tilde{f} = \alpha f$, then
\begin{equation*}
  -\im (\partial_t + \im B) \rho \tilde{u} = \iota_2 \tilde{f}.
\end{equation*}

The projection $\pi_2$ and the embeddings $\rho$, $\iota_2$, which relate solutions of the Klein--Gordon equation and the first-order Klein--Gordon equation, can be understood between various spaces.
It follows from the definition of $\mathcal{H}_\lambda$ in~\eqref{eq:H_alpha} that, for $\lambda \in [-1,1]$,
\begin{subequations}\label{eq:pi2-iota2}\begin{align}
  \pi_2 &: \mathcal{H}_\lambda \to \mathcal{K}^{(\lambda-1)/2}, \\
  \pi_2 Q &: \mathcal{H}_\lambda \to \mathcal{K}^{(\lambda+1)/2}, \\
  \iota_2 &: \mathcal{K}^{(\lambda-1)/2} \to \mathcal{H}_\lambda.
\end{align}\end{subequations}

These projections and embeddings already allow us to easily prove an existence and uniqueness result regarding solutions of the Klein--Gordon equation with Cauchy data in the energy space:
\begin{theorem}
  Let $s \in \RR$, $\,\left(\!\begin{smallmatrix} u_1(s) \\ u_2(s) \end{smallmatrix}\!\right) \in \Hen$ and $f \in L^1_\loc(\RR; \mathcal{K}^0)$.
  Set
  \begin{equation*}
    \begin{pmatrix} \tilde{u}_1(s) \\ \tilde{u}_2(s) \end{pmatrix} = \alpha(s)^{-1} \begin{pmatrix} u_1(s) \\ u_2(s) \end{pmatrix}
    \quad\text{and}\quad
    \tilde{f} = \alpha f.
  \end{equation*}
  Then $u = \alpha \tilde{u}$ with
  \begin{equation*}
    \tilde{u}(t) = \pi_2 Q U(t,s) \begin{pmatrix} \tilde{u}_1(s) \\ \tilde{u}_2(s) \end{pmatrix} + \im \int_s^t \pi_2 Q U(t,r) \iota_2 \tilde{f}(r)\, \dif r
  \end{equation*}
  is the unique solution of $\KG u=f$ such that
  \begin{equation}\label{eq:a5}
    u \in C(\RR; \mathcal{K}^1) \cap C^1(\RR; \mathcal{K}^0)
    \quad\text{and}\quad
    \rho \tilde{u}(s) = \begin{pmatrix} \tilde{u}_1(s) \\ \tilde{u}_2(s) \end{pmatrix}.
  \end{equation}
\end{theorem}
\begin{proof}
  We have the following special cases of~\eqref{eq:pi2-iota2}:
  \begin{subequations}\begin{align}
    \iota_2 &: \mathcal{K}^0 \to \Hen, \label{eq:iota-prop} \\
    \pi_2 Q &: \Hen \to \mathcal{K}^1, \label{eq:piQ-prop1} \\
    \pi_2 Q &: \Hne \to \mathcal{K}^0. \label{eq:piQ-prop2}
  \end{align}\end{subequations}
  By~\eqref{eq:iota-prop}, \eqref{eq:piQ-prop1} and Assumption~\ref{asm:beta_cont}, $u$ belongs to $C(\RR; \mathcal{K}^1)$.
  By~\eqref{eq:iota-prop}, \eqref{eq:piQ-prop2} and Assumption~\ref{asm:beta_cont}, $\partial_t u$ belongs to $C(\RR; \mathcal{K}^0)$.
  Hence the first part of~\eqref{eq:a5} is true.
  The second part of~\eqref{eq:a5} is obvious.

  Set
  \begin{equation}\label{eq:a6}
    \begin{pmatrix} \tilde{u}_1(t) \\ \tilde{u}_2(t) \end{pmatrix} = U(t,s) \begin{pmatrix} \tilde{u}_1(s) \\ \tilde{u}_2(s) \end{pmatrix} + \im \int_s^t U(t,r) \iota_2 \tilde{f}(r)\, \dif r.
  \end{equation}
  Differentiating~\eqref{eq:a6} we obtain
  \begin{equation}\label{eq:a7}
    \im\partial_t \begin{pmatrix} \tilde{u}_1(t) \\ \tilde{u}_2(t) \end{pmatrix} = B(t) \begin{pmatrix} \tilde{u}_1(t) \\ \tilde{u}_2(t) \end{pmatrix} - \iota_2 \tilde{f}(t)
  \end{equation}
  Clearly, $\tilde{u}(t)=\tilde{u}_1(t)$.
  The first component of~\eqref{eq:a7} yields $\tilde{u}_2(t) = -(D_t + W(t)) \tilde{u}_1(t)$.
  Hence
  \begin{equation}\label{eq:a8}
    \rho \tilde{u}(t) = \begin{pmatrix} \tilde{u}_1(t) \\ \tilde{u}_2(t) \end{pmatrix}
  \end{equation}
  The second component of~\eqref{eq:a7}, and then insertion of~\eqref{eq:a8} yield
  \begin{equation*}
    \tilde{f}(t) = -\im \pi_2 (\partial_t + \im B) \begin{pmatrix} \tilde{u}_1(t) \\ \tilde{u}_2(t) \end{pmatrix} = -\im \pi_2 (\partial_t + \im B)\rho \tilde{u}(t)
    = \tilde\KG \tilde{u}(t),
  \end{equation*}
  whence we have shown that $\tilde{u}$ solves $\tilde\KG \tilde{u} = \tilde{f}$ and thus $\KG u = f$.

  Uniqueness of the solution follows from the uniqueness of the evolution~$U(t,s)$, and the linearity of $\KG, \rho$ by the standard argument:
  If $u,u'$ satisfy
  \begin{equation*}
    \KG u = \KG u' = f
    \quad\text{and}\quad
    \rho \tilde{u}(s) = \rho \tilde{u}'(s) = \begin{pmatrix} \tilde{u}_1(s) \\ \tilde{u}_2(s) \end{pmatrix},
  \end{equation*}
  where $\tilde{u}' = \alpha^{-1} u'$, then $\KG(u-u') = 0$, $\rho (\tilde{u}-\tilde{u}')(s) = 0$ and thus $u=u'$.
\end{proof}

It is well-known that solutions of the Klein--Gordon equation propagate slower than the speed of light.
The method of evolution equations together with the freedom of the choice of the time-variable provide a rather obvious heuristic argument for the propagation at a finite speed.
However, when one tries to convert this argument into a rigorous proof, technical problems appear which make such a proof difficult to formulate.

In the literature the finiteness of the speed of propagation is usually shown for the Klein--Gordon equation with smooth coefficients.
In Appx.~\ref{appx:finite-speed}, in particular in Thm.~\ref{thm:finite-speed}, we show that solutions of the Klein--Gordon propagate at a finite speed also in a low-regularity setup typical for our paper.

\section{Classical propagators}
\label{sec:classical}

Having constructed the evolution for~$B(t)$ in Sect.~\ref{sec:evolution}, it is not difficult to find the classical propagators for the first-order Klein--Gordon operator $\partial_t + \im B$.
To wit, the \emph{Pauli--Jordan propagator}~$E^\PJ$ and the \emph{forward/backward propagator}~$E^\retadv$ are given by the integral kernels
\begin{subequations}\label{eq:classical_kernels}\begin{align}
    E^\PJ(t, s) &\defn U(t, s), \\
    E^\vee(t, s) &\defn \theta(t - s)\, U(t, s), \\
    E^\wedge(t, s) &\defn -\theta(s - t)\, U(t, s),
\end{align}\end{subequations}
where $\theta$ denotes the Heaviside step function, via
\begin{equation}\label{eq:E_op-E_kernel}
  (E^\bullet f)(t) = \int_\RR E^\bullet(t, s) f(s)\, \dif s.
\end{equation}

\begin{theorem}\label{thm:E_classical}
  Let $\lambda \in [-1,1]$.
  \begin{enumerate}
    \item\label{item.i}
      The classical propagators $E^\PJ$ and $E^\retadv$ are well-defined between the following spaces:
      \begin{align*}
        E^\bullet &: L^1_\comp(\RR; \mathcal{H}_\lambda) \to C(\RR; \mathcal{H}_\lambda), \\
        E^\bullet &: L^1_\comp(\RR; \Hen) \to C^1(\RR; \Hne).
      \end{align*}
    \item\label{item.ii}
      The forward and backward propagator $E^\retadv$ are well-defined between the following spaces:
      \begin{align*}
        E^\retadv &: L^1_\loc(I; \mathcal{H}_\lambda) \to C(I; \mathcal{H}_\lambda), \notag\\
        E^\retadv &: L^1_\loc(I; \Hen) \to C^1(I; \Hne),
      \end{align*}
      where $I = [a,+\infty\mathclose{[}$ resp. $\mathopen{]}-\infty,a]$ for some $a \in \RR$.
    \item\label{item.iii}
      If Assumption~\ref{asm:asymptotic} is satisfied, the classical propagators $E^\PJ$ and $E^\retadv$ are bounded between the following spaces:
      \begin{align*}
        E^\bullet &: L^1(\RR; \mathcal{H}_\lambda) \to C_\bnd(\RR; \mathcal{H}_\lambda), \\
        E^\bullet &: L^1(\RR; \Hen) \to C_\bnd^1(\RR; \Hne).
      \end{align*}
    \item\label{item.iv}
      $E^\PJ$ is a bisolution of $\partial_t + \im B$:
      \begin{align}
        (\partial_t + \im B) E^\PJ f &= 0,
        \quad
        f \in L_\comp^1(\RR;\Hen), \label{eq:b2}
        \\
        E^\PJ (\partial_t + \im B) f &= 0,
        \quad
        f \in L^1_\comp(\RR; \Hen) \cap AC_\comp(\RR;\Hne). \label{eq:b1}
      \end{align}
    \item\label{item.v}
      $E^\retadv$ are the unique inverses of $\partial_t + \im B$ such that
      \begin{align}
        (\partial_t + \im B) E^\retadv f &= f,
        \quad
        f \in L_\loc^1(I,\Hen), \label{eq:b4}
        \\
        E^\retadv (\partial_t + \im B) f &= f,
        \quad
        f \in L^1_\loc(I; \Hen)\cap AC(I,\Hne), \label{eq:b3}
      \end{align}
       with $I = [a,+\infty\mathclose{[}$ resp. $\mathopen{]}-\infty,a]$ for some $a \in \RR$.
    \item The relation $E^\PJ = E^\vee - E^\wedge$ holds.
  \end{enumerate}
\end{theorem}
\begin{proof}
  \ref{item.i}--\ref{item.iii} follow from the properties of the evolution $U(t, s)$ (see Thm.~\ref{thm:evolution} and Cor.~\ref{cor:evolution-bound}) and the definition of the kernels~\eqref{eq:classical_kernels}.

  Consider next~\ref{item.iv} and~\ref{item.v}.
  We first need to check that the products contained in these properties are well-defined.
  Indeed, by~\ref{item.i}, the following maps are well-defined:
  \begin{subequations}\label{eq:E-KG_map1}\begin{align}
    E^\bullet &: L_\comp^1(\RR;\Hen) \to C(\RR;\Hen)\cap C^1(\RR;\Hne), \\
    (\partial_t+\im B) &: C(\RR;\Hen) \cap C^1(\RR;\Hne) \to C(\RR;\Hne),
  \end{align}\end{subequations}
  which shows that~\eqref{eq:b2} and~\eqref{eq:b4} make sense.
  Similarly, by~\ref{item.i}, we have
  \begin{subequations}\label{eq:E-KG_map2}\begin{align}
    (\partial_t+\im B) &: L_\comp^1(\RR;\Hen) \cap AC_\comp(\RR;\Hne) \to L_\comp^1(\RR;\Hne), \\
    E^\bullet &: L_\comp^1(\RR;\Hne) \to C(\RR;\Hne),
  \end{align}\end{subequations}
  hence the products in~\eqref{eq:b1} and~\eqref{eq:b3} make sense.
  Then we show \eqref{eq:b2}--\eqref{eq:b3} using~\eqref{eq:E_op-E_kernel} and~\eqref{eq:U_sol}.
  For~\eqref{eq:b1} and~\eqref{eq:b3} we also need to apply an integration by parts.
\end{proof}

We can also state an $L^2$ version of~\ref{item.iii} in Thm.~\ref{thm:E_classical} above:
\begin{theorem}\label{thm:E_classical-L2}
  Let $s > \tfrac12$ and $\lambda \in [-1,1]$.
  If Assumption~\ref{asm:asymptotic} is satisfied, the classical propagators $E^\PJ$ and $E^\retadv$ are bounded between the following spaces:
  \begin{align*}
    E^\bullet &: \jnorm{t}^{-s} L^2(\RR; \mathcal{H}_\lambda) \to \jnorm{t}^s L^2(\RR; \mathcal{H}_\lambda), \\
    E^\bullet &: \jnorm{t}^{-s} L^2(\RR; \Hen) \to \jnorm{t}^s\jnorm{\partial_t}^{-1} L^2(\RR; \Hne).
  \end{align*}
\end{theorem}
\begin{proof}
  We use the embeddings
  \begin{equation*}
    \jnorm{t}^{-s} L^2(\RR; \mathcal{X}) \subset L^1(\RR; \mathcal{X})
    \quad\text{and}\quad
    \jnorm{t}^s L^2(\RR; \mathcal{X}) \supset C_\bnd(\RR; \mathcal{X})
  \end{equation*}
  for any Banach space $\mathcal{X}$ and $s > \frac12$.
\end{proof}

The classical propagators for the first-order Klein--Gordon operator can also be understood between various spaces other than those considered in Thms.~\ref{thm:E_classical} and~\ref{thm:E_classical-L2}, but our choices are quite natural.
At the same time, this setup leads to an almost straightforward derivation of the propagators for the Klein--Gordon operator~$\KG$.

Since $\partial_t + \im B$ and $\KG$ are related via~\eqref{eq:KG1-KG}, also the propagators of these operators are closely related.
At least formally, it can be shown that if $E^\bullet$ is a propagator for $\partial_t + \im B$, then $\im \pi_2 Q E^\bullet \iota_2$ is a propagator for $\tilde\KG$, and hence
\begin{equation}\label{eq:G-E}
  G^\bullet = \im \alpha \pi_2 Q E^\bullet \iota_2 \alpha.
\end{equation}
is a propagator for the Klein--Gordon operator~$\KG$.
As we shall see now, this is indeed true if the domain of $G^\bullet$ is carefully chosen:
\begin{theorem}\label{thm:G_classical}
  Let $\delta \in [0,1]$.
  \begin{enumerate}
    \item\label{item-i}
      The classical propagators $G^\PJ$ and $G^\retadv$ are well-defined between the following spaces:
      \begin{align}
        G^\bullet &: L^1_\comp(\RR; \mathcal{K}^{-\delta}) \to C(\RR; \mathcal{K}^{1-\delta}), \label{eq:w1}\\
        G^\bullet &: L^1_\comp(\RR; \mathcal{K}^0) \to C^1(\RR; \mathcal{K}^0). \label{eq:w2}
      \end{align}
    \item\label{item-ii}
      The forward and backward propagators $G^\retadv$ are well-defined between the following spaces:
      \begin{align*}
        G^\retadv &: L^1_\loc(I; \mathcal{K}^{-\delta}) \to C(I; \mathcal{K}^{1-\delta}), \notag\\
        G^\retadv &: L^1_\loc(I; \mathcal{K}^0) \to C^1(I; \mathcal{K}^0),
      \end{align*}
      where $I = [a,+\infty\mathclose{[}$ resp. $\mathopen{]}-\infty,a]$ for some $a \in \RR$.
    \item\label{item-iii}
      If Assumption~\ref{asm:asymptotic} is satisfied, the classical propagators $G^\PJ$ and $G^\retadv$ are bounded between the following spaces:
      \begin{align*}
        G^\bullet &: L^1(\RR; \mathcal{K}^{-\delta}) \to C_\bnd(\RR; \mathcal{K}^{1-\delta}), \\
        G^\bullet &: L^1(\RR; \mathcal{K}^0) \to C_\bnd^1(\RR; \mathcal{K}^0).
      \end{align*}
    \item\label{item-iv}
      $G^\PJ$ is a bisolution of $\KG$:
      \begin{align}
        \KG G^\PJ f &= 0,
        \quad
        f \in L^1_\comp(\RR; \mathcal{K}^0),\label{eq:c1}
        \\
        G^\PJ \KG f &= 0,
        \quad
        f \in L^1_\comp(\RR; \mathcal{K}^1) \cap AC_\comp(\RR;\mathcal{K}^{0}) \cap AC_\comp^1(\RR;\mathcal{K}^{-1}). \label{eq:c2}
      \end{align}
    \item\label{item-v}
      $G^\retadv$ are the unique inverses of $\KG$ such that
      \begin{align}
        \KG G^\retadv f &= f,
        \quad
        f \in L^1_\loc(I; \mathcal{K}^0), \label{eq:c3}
        \\
        G^\retadv \KG f &= f,
        \quad
        f \in L^1_\loc(I; \mathcal{K}^1) \cap AC(\RR;\mathcal{K}^{0}) \cap AC^1(I;\mathcal{K}^{-1}). \label{eq:c4}
      \end{align}
      with $I = [a,+\infty\mathclose{[}$ resp. $\mathopen{]}-\infty,a]$ for some $a \in \RR$.
    \item The relation $G^\PJ = G^\vee - G^\wedge$ holds.
  \end{enumerate}
\end{theorem}
\begin{proof}
  These results are a direct consequence of Thm.~\ref{thm:E_classical}.
  In \ref{item-i}--\ref{item-iii} we used~\eqref{eq:pi2-iota2} and Assumption~\ref{asm:beta_cont}.

  Let us check that the products in~\ref{item-iv} and~\ref{item-v} are well-defined.
  From the definition of $\rho$ we can read off that
  \begin{align*}
    \rho &: C(\RR; \mathcal{K}^1) \cap C^1(\RR; \mathcal{K}^0) \to C(\RR; \Hen), \\
    \rho &: L^1_\comp(\RR; \mathcal{K}^1) \cap AC_\comp(\RR; \mathcal{K}^0) \to L^1_\comp(\RR; \Hen), \\
    \rho &: AC_\comp(\RR; \mathcal{K}^0) \cap AC^1_\comp(\RR; \mathcal{K}^{-1}) \to AC_\comp(\RR; \Hne).
  \end{align*}
  Then, by~\ref{item-i} and also using~\eqref{eq:E-KG_map1}, we have
  \begin{align*}
    G^\bullet &: L_\comp^1(\RR;\mathcal{K}^0) \to C(\RR;\mathcal{K}^1) \cap C^1(\RR;\mathcal{K}^0), \\
    \KG &: C(\RR;\mathcal{K}^1) \cap C^1(\RR;\mathcal{K}^0) \to C^{-1}(\RR;\mathcal{K}^0) \cap C(\RR;\mathcal{K}^{-1}),
  \end{align*}
  where $C^{-1}(\RR)$ denotes the space of distributional derivatives of continuous functions.
  This shows that~\eqref{eq:c1} and~\eqref{eq:c3} make sense.
  Similarly, by~\ref{item-i} and~\eqref{eq:E-KG_map2}, we have
  \begin{align*}
    \KG &: L_\comp^1(\RR;\mathcal{K}^1) \cap AC_\comp(\RR;\mathcal{K}^0) \cap AC_\comp^1(\RR;\mathcal{K}^{-1}) \to L_\comp^1(\RR;\mathcal{K}^{-1}),\\
    G^\bullet &: L_\comp^1(\RR;\mathcal{K}^{-1}) \to C(\RR;\mathcal{K}^{0}),
  \end{align*}
  hence the products in~\eqref{eq:c2} and~\eqref{eq:c4} make sense.
\end{proof}

Here is an $L^2$ version of~\ref{item-iii} in Thm.~\ref{thm:G_classical}:
\begin{theorem}
  Let $s > \tfrac12$.
  If Assumption~\ref{asm:asymptotic} is satisfied, the classical propagators~$G^\PJ$ and~$G^\retadv$ are bounded between the following spaces:
  \begin{subequations}\begin{align}
    G^\bullet &: \jnorm{t}^{-s} L^2(\Omega^\frac12 M) \to \jnorm{t}^s L(t)^{-\frac12} L^2(\Omega^\frac12 M),\label{eq:q1}\\
    G^\bullet &: \jnorm{t}^{-s} L^2(\Omega^\frac12 M) \to \jnorm{t}^s \jnorm{\partial_t}^{-1} L^2(\Omega^\frac12 M).\label{eq:q2}
  \end{align}\end{subequations}
\end{theorem}
\begin{proof}
  By~\eqref{eq:w1}, for $\delta \in [0,1]$ we have
  \begin{equation}\label{eq:q3}
    G^\bullet : \jnorm{t}^{-s} L^2(\RR; \mathcal{K}^{-\delta}) \to \jnorm{t}^s L^2(\RR; \mathcal{K}^{1-\delta}).
  \end{equation}
  Setting $\delta=0$ we obtain
  \begin{equation}\label{eq:q4}
    G^\bullet : \jnorm{t}^{-s} L^2(\RR; \mathcal{K}^0) \to \jnorm{t}^s L(t)^{-\frac12} L^2(\RR; \mathcal{K}^0).
  \end{equation}
  But $L^2(\RR; \mathcal{K}^0) = L^2(\RR; L^2(\Omega^\frac12 \Sigma))$ and $L^2(\Omega^\frac12 M)$ can naturally be identified, which proves~\eqref{eq:q1}.

  It follows from~\eqref{eq:w2} that
  \begin{equation*}
    G^\bullet:\jnorm{t}^{-s} L^2(\RR; \mathcal{K}^0) \to \jnorm{t}^s \jnorm{\partial_t}^{-1} L^2(\RR; \mathcal{K}^0).
  \end{equation*}
  This yields~\eqref{eq:q2}.
  \end{proof}

Observe that in other approaches, \eg~\cite{bar}, the retarded and advanced propagators are the central objects and the Pauli--Jordan propagator is defined as their difference.
Here, instead, the Pauli--Jordan propagator follows immediately from the evolution $U(t,s)$ and should be seen as the central object, while the retarded and advanced propagators are derived objects.

Using the Pauli--Jordan propagator $G^\PJ$, we can associate to every sufficiently regular compactly supported function a solution of the homogeneous Klein--Gordon equation.
In fact, as the following proposition shows, also the converse is true.
\begin{proposition}
  Suppose that $u \in L_\loc^1(\RR;\mathcal{K}^1) \cap AC(\RR;\mathcal{K}) \cap AC^1(\RR; \mathcal{K}^{-1})$ satisfies $\KG u=0$.
  Then there exists a (non-unique) $f \in L_\comp^1(\RR; \mathcal{K}^{-1})$ such that $u = G^\PJ f$.
\end{proposition}
\begin{proof}
  Choose any $r,s \in \RR$, $r < s$, and $\chi \in C^\infty(M)$ such that $\chi(t) = 0$ for $t < r$, $0 \leq \chi(t) \leq 1$ for $r \leq t \leq s$ and $\chi(t) = 1$ for $t > s$.
  Clearly,
  \begin{equation*}
    0 = \KG u = \KG \chi u - \KG (\chi-1) u
  \end{equation*}
  and thus $\supp(\KG \chi u) \subset [r,s] \times \Sigma$.
  Besides, $\KG\chi u\in L_\comp^1(\RR;\mathcal{K}^{-1})$.
  Therefore, we can act with $G^\PJ$ on $\KG\chi u$, obtaining
  \begin{equation*}
    G^\PJ \KG \chi u = G^\vee \KG \chi u - G^\wedge \KG (\chi-1) u = u.
  \end{equation*}
  That is, $f = \KG \chi u$ is the desired function.
\end{proof}

Our construction of the classical propagators starts from the propagators for the first-order Klein--Gordon operator (\ie, given $E^\bullet$, we derive $G^\bullet$ using~\eqref{eq:G-E}).
If, instead, $G^\bullet$ is provided, then $E^\bullet$ can be derived:
\begin{enumerate}
  \item
    If $G^\bullet$ is an inverse of $\KG$ then
    \begin{equation*}
      E^\bullet = -\im \begin{pmatrix}
        -\alpha^{-1} G^\bullet \alpha^{-1} (D_t + W^*) & \alpha^{-1} G^\bullet \alpha^{-1} \\
        \one + (D_t + W) \alpha^{-1} G^\bullet \alpha^{-1} (D_t + W^*) & -(D_t + W) \alpha^{-1} G^\bullet \alpha^{-1} \\
      \end{pmatrix}
    \end{equation*}
    is (formally) an inverse of $(\partial_t + \im B)$.
  \item
    If $G^\bullet$ is a bisolution of $\KG$ then
    \begin{equation*}
      E^\bullet = -\im \begin{pmatrix}
        -\alpha^{-1} G^\bullet \alpha^{-1} (D_t + W^*) & \alpha^{-1} G^\bullet \alpha^{-1} \\
        (D_t + W) \alpha^{-1} G^\bullet \alpha^{-1} (D_t + W^*) & -(D_t + W) \alpha^{-1} G^\bullet \alpha^{-1} \\
      \end{pmatrix}
    \end{equation*}
    is (formally) a bisolution of $(\partial_t + \im B)$.
\end{enumerate}
Note the subtle difference in the formulas for inverses and bisolutions.
No such difference appears in~\eqref{eq:G-E} which yields $G^\bullet$ given $E^\bullet$.

\section{Instantaneous non-classical propagators}
\label{sec:instantaneous}

Consider an arbitrary reference time $\tau$.
According to Prop.~\ref{prop:B}, $B(\tau)$ is a self-adjoint operator on~$\Hne[\tau]$.
Therefore we can use the functional calculus to define the projections onto the positive and negative spectrum of~$B(\tau)$:
\begin{equation}\label{eq:comple}
  \Pi_\tau ^{(\pm)} \defn \one_{[0,\infty[}\bigl(\pm B(\tau)\bigr).
\end{equation}
Zero is in the resolvent set of $B(\tau)$, and therefore \eqref{eq:comple} are complementary.
\begin{proposition}\label{prps1}
  $\Pi^{(\pm)}_\tau$ restrict to complementary projections on $\mathcal{H}_\lambda$ for $\lambda \in [-1,1]$, and have the following properties:
  \begin{enumerate}
    \item $\Pi^{(\pm)}_\tau B(\tau) = B(\tau) \Pi^{(\pm)}_\tau $,
    \item $\Pi^{(+)}_\tau - \Pi^{(-)}_\tau = \sgn B(\tau)$,
    \item $\spec \bigl( \pm\Pi^{(\pm)}_\tau B(\tau) \bigr) \subset \mathopen{]}0,\infty\mathclose{[}$,
    \item $\Pi^{(\pm)}_\tau $ are self-adjoint with respect to $\mathcal{H}_{\lambda,\tau}$.
  \end{enumerate}
\end{proposition}

Moreover, the projections $\Pi^{(\pm)}_\tau $ split $\mathcal{H}_{\lambda,\tau}$ into subspaces of positive and negative charge (with respect to the charge form $Q$):
\begin{proposition}\label{prop:Q_pos}
  \begin{equation}
    \pm (u \,|\, Q \Pi^{(\pm)}_\tau u)
    = \pm(\Pi^{(\pm)}_\tau u \,|\, Q u)
    = \pm(\Pi^{(\pm)}_\tau u \,|\, Q \Pi^{(\pm)}_\tau u)
    \geq 0
  \end{equation}
  for all $u \in \mathcal{H}_\lambda$ with $\lambda \in [-1,1]$.
\end{proposition}
\begin{proof}
  The proof is the same as for Prop.~6.3 in~\cite{derezinski-siemssen}.
\end{proof}

The projections $\Pi^{(\pm)}_\tau$ can be used to define \emph{instantaneous positive/negative frequency bisolutions} $E^{(\pm)}_\tau$, given by their integral kernels as
\begin{equation}\label{eq:kernel_posneg}
  E^{(\pm)}_\tau(t, s) \defn \pm U(t, \tau) \Pi^{(\pm)}_\tau U(\tau, s).
\end{equation}
Using step functions, we then define the kernels of the \emph{instantaneous Feynman and anti-Feynman inverses of} $\partial_t + \im B$:
\begin{align*}
  E^\Feyn_\tau(t, s)  &\defn \theta(t - s) E^{(+)}_\tau(t, s) + \theta(s - t) E^{(-)}_\tau(t, s), \\
  E^\aFeyn_\tau(t, s) &\defn -\theta(t - s) E^{(-)}_\tau(t, s) - \theta(s - t) E^{(+)}_\tau(t, s).
\end{align*}
It is easy to see that these kernels can also be expressed using the retarded and advanced propagators:
\begin{subequations}\label{eq:E_qkernel_rels}\begin{align}
  E^\Feyn_\tau(t, s) &= E^\wedge(t, s) + E^{(+)}_\tau(t, s) = E^\vee(t, s) + E^{(-)}_\tau(t, s), \\
  E^\aFeyn_\tau(t, s) &= E^\vee(t, s) - E^{(+)}_\tau(t, s) = E^\wedge(t, s) - E^{(-)}_\tau(t, s).
\end{align}\end{subequations}

As before, these kernels define the corresponding propagators via~\eqref{eq:E_op-E_kernel}:
\begin{theorem}\label{thm:E_quantum}
  Let $\lambda \in [-1,1]$.
  \begin{enumerate}
    \item
      The instantaneous non-classical propagators $E^{(\pm)}_\tau$ and $E^\FeynFeyn_\tau$ are well-defined between the following spaces:
      \begin{align*}
        E^\bullet_\tau &: L^1_\comp(\RR; \mathcal{H}_\lambda) \to C(\RR; \mathcal{H}_\lambda), \\
        E^\bullet_\tau &: L^1_\comp(\RR; \Hen) \to C^1(\RR; \Hne).
      \end{align*}
    \item \label{item:E_quantum:2}
      If Assumption~\ref{asm:asymptotic} is satisfied, $E^{(\pm)}_\tau$ and $E^\FeynFeyn_\tau$ are bounded between the following spaces:
      \begin{align*}
        E^\bullet_\tau &: L^1(\RR; \mathcal{H}_\lambda) \to C_\bnd(\RR; \mathcal{H}_\lambda), \\
        E^\bullet_\tau &: L^1(\RR; \Hen) \to C^1_\bnd(\RR; \Hne).
      \end{align*}
    \item
      $E^{(\pm)}_\tau$ are bisolutions of $\partial_t + \im B$:
      \begin{align*}
        (\partial_t + \im B) E^{(\pm)}_\tau f &= 0,
        \quad
        f \in L^1_\comp(\RR; \Hen),
        \\
        E^{(\pm)}_\tau (\partial_t + \im B) f &= 0,
        \quad
        f \in L^1_\comp(\RR; \Hen) \cap AC_\comp(\RR;\Hne).
      \end{align*}
    \item
      $E^\FeynFeyn_\tau$ are inverses of $\partial_t + \im B$:
      \begin{align*}
        (\partial_t + \im B) E^\FeynFeyn_\tau f &= f,
        \quad
        f \in L^1_\comp(\RR; \Hen),
        \\
        E^\FeynFeyn_\tau (\partial_t + \im B) f &= f,
        \quad
        f \in L^1_\comp(\RR; \Hen) \cap AC_\comp(\RR;\Hne).
      \end{align*}
    \item \label{item:E_quantum:5}
      The instantaneous non-classical propagators satisfy the relations:
      \begin{alignat*}{3}
        E^\Feyn_\tau &= E^\wedge \mkern-1mu + E^{(+)}_\tau = E^\vee \mkern-1mu + E^{(-)}_\tau,
        &\qquad
        E^\Feyn_\tau + E^\aFeyn_\tau &= E^\vee \mkern-1mu + E^\wedge,
        &\qquad
        E^{(+)}_\tau \mkern-1mu - E^{(-)}_\tau = E^\PJ, \\
        E^\aFeyn_\tau &= E^\vee \mkern-1mu - E^{(+)}_\tau = E^\wedge \mkern-1mu - E^{(-)}_\tau,
        &\qquad
        E^\Feyn_\tau - E^\aFeyn_\tau &= E^{(+)}_\tau \mkern-1mu + E^{(-)}_\tau.
      \end{alignat*}
  \end{enumerate}
\end{theorem}
\begin{proof}
  The various properties of the non-classical propagators can be shown along the same lines as in Thm.~\ref{thm:E_classical} so we will omit the proofs.
  Property~\ref{item:E_quantum:5} in particular follows from~\eqref{eq:E_qkernel_rels} and its linear combinations.
\end{proof}

As for the classical propagators, we can also find an $L^2$ version of~\ref{item:E_quantum:2} in Thm.~\ref{thm:E_quantum}:
\begin{theorem}
  Let $s > \tfrac12$ and $\lambda \in [-1,1]$.
  If Assumption~\ref{asm:asymptotic} is satisfied, the instantaneous non-classical propagators $E^{(\pm)}_\tau$ and $E^\FeynFeyn_\tau$ are bounded between the following spaces:
  \begin{align*}
    E^\bullet_\tau &: \jnorm{t}^{-s} L^2(\RR; \mathcal{H}_\lambda) \to \jnorm{t}^s L^2(\RR; \mathcal{H}_\lambda), \\
    E^\bullet_\tau &: \jnorm{t}^{-s} L^2(\RR; \Hen) \to \jnorm{t}^s\jnorm{\partial_t}^{-1} L^2(\RR; \Hne).
  \end{align*}
\end{theorem}

Similar to~\eqref{eq:G-E}, we define the instantaneous non-classical propagators $G^{(\pm)}_\tau$ and $G^\FeynFeyn_\tau$ for the Klein--Gordon operator~$\KG$ by
\begin{align*}
  G^{(\pm)}_\tau \defn \alpha \pi_2 Q E^{(\pm)}_\tau \iota_2 \alpha,
  \quad
  G^\FeynFeyn_\tau \defn \im \alpha \pi_2 Q E^\FeynFeyn_\tau \iota_2 \alpha.
\end{align*}
Note the absence of the complex unit in the definition of $G^{(\pm)}_\tau$ so that~$G^{(\pm)}_\tau$ define positive forms, see property~\ref{item:G_quantum:6} below.

Analogously to Thm.~\ref{thm:G_classical}, we find
\begin{theorem}\label{thm:G_quantum}
  Let $\delta \in [-1,1]$.
  \begin{enumerate}
    \item
      The instantaneous non-classical propagators $G^{(\pm)}_\tau$ and $G^\FeynFeyn_\tau$ are well-defined between the following spaces:
      \begin{align*}
        G^\bullet_\tau &: L^1_\comp(\RR; \mathcal{K}^{-\delta}) \to C(\RR; \mathcal{K}^{1-\delta}), \\
        G^\bullet_\tau &: L^1_\comp(\RR; \mathcal{K}^0) \to C^1(\RR; \mathcal{K}^0).
      \end{align*}
    \item \label{item:G_quantum:2}
      If Assumption~\ref{asm:asymptotic} is satisfied, $G^{(\pm)}_\tau$ and $G^\FeynFeyn_\tau$ are bounded between the following spaces:
      \begin{align*}
        G^\bullet_\tau &: L^1(\RR; \mathcal{K}^{-\delta}) \to C_\bnd(\RR; \mathcal{K}^{1-\delta}), \\
        G^\bullet_\tau &: L^1(\RR; \mathcal{K}^0) \to C^1_\bnd(\RR; \mathcal{K}^0).
      \end{align*}
    \item
      $G^{(\pm)}_\tau$ are bisolutions of $\KG$:
      \begin{align*}
        \KG G^{(\pm)}_\tau f &= 0,
        \quad
        f \in L^1_\comp(\RR; \mathcal{K}^0),
        \\
        G^{(\pm)}_\tau \KG f &= 0,
        \quad
        f \in L^1_\comp(\RR; \mathcal{K}^1) \cap AC_\comp(\RR;\mathcal{K}^0) \cap AC^1_\comp(\RR; \mathcal{K}^{-1}).
      \end{align*}
    \item
      $G^\FeynFeyn_\tau$ are inverses of $\KG$:
      \begin{align*}
        \KG G^\FeynFeyn_\tau f &= f,
        \quad
        f \in L^1_\comp(\RR; \mathcal{K}^0),
        \\
        G^\FeynFeyn_\tau \KG f &= f,
        \quad
        f \in L^1_\comp(\RR; \mathcal{K}^1) \cap AC_\comp(\RR;\mathcal{K}^0) \cap AC^1_\comp(\RR; \mathcal{K}^{-1}).
      \end{align*}
    \item \label{item:G_quantum:5}
      The instantaneous non-classical propagators satisfy the relations:
      \begin{alignat*}{3}
        G^\Feyn_\tau &= G^\wedge \mkern-1mu + \im G^{(+)}_\tau = G^\vee \mkern-1mu + \im G^{(-)}_\tau,
        &\qquad
        G^\Feyn_\tau + G^\aFeyn_\tau &= G^\vee \mkern-1mu + G^\wedge,
        &\qquad
        G^{(+)}_\tau \mkern-1mu - G^{(-)}_\tau = -\im G^\PJ, \\
        G^\aFeyn_\tau &= G^\vee \mkern-1mu - \im G^{(+)}_\tau = G^\wedge \mkern-1mu - \im G^{(-)}_\tau,
        &\qquad
        G^\Feyn_\tau - G^\aFeyn_\tau &= \im G^{(+)}_\tau \mkern-1mu + \im G^{(-)}_\tau.
      \end{alignat*}
    \item \label{item:G_quantum:6}
      The instantaneous positive/negative frequency bisolutions are positive:
      \begin{equation*}
        (f \,|\, G^{(\pm)}_\tau f) = \int_M \conj{f}\, G^{(\pm)}_\tau f \geq 0
      \end{equation*}
      for $f \in L^1_\comp(\RR; \mathcal{K}^0)$.
  \end{enumerate}
\end{theorem}
\begin{proof}
  We only show~\ref{item:G_quantum:6}; the remaining properties follow from corresponding properties of~$E^\bullet_\tau$ in Thm.~\ref{thm:E_quantum} and can be shown as in Thm.~\ref{thm:G_classical}.
  For~\ref{item:G_quantum:6}, we note that
  \begin{align*}
    (f \,|\, G^{(\pm)}_\tau f)
    &= \iint \bigl( \iota_2 \tilde{f}(t) \,\big|\, Q E^{(\pm)}_\tau(t,s) \iota_2 \tilde{f}(s) \bigr)\, \dif s\, \dif t \\
    &= \bigl( \tilde{u}(\tau) \,\big|\, Q \Pi^{(\pm)}_\tau \tilde{u}(\tau) \bigr)
    \geq 0
  \end{align*}
  by Prop.~\ref{prop:Q_pos}, where we set $\tilde{f} = \alpha f$ and $\tilde{u}(\tau) = \int U(\tau, t) \tilde{f}(t)\, \dif t \in \Hen$.
\end{proof}

The $L^2$ version of~\ref{item:G_quantum:2} of Thm.~\ref{thm:G_quantum} is:
\begin{theorem}
  Let $s > \tfrac12$.
  If Assumption~\ref{asm:asymptotic} is satisfied, the instantaneous non-classical propagators~$G^{(\pm)}_\tau$ and~$G^\FeynFeyn_\tau$ are bounded between the following spaces:
  \begin{align*}
    G^\bullet_\tau &: \jnorm{t}^{-s} L^2(\Omega^\frac12 M) \to \jnorm{t}^sL(t)^{-\frac12} L^2(\Omega^\frac12 M),\\
    G^\bullet_\tau &: \jnorm{t}^{-s} L^2(\Omega^\frac12 M) \to \jnorm{t}^s\jnorm{\partial_t}^{-1} L^2(\Omega^\frac12 M).
  \end{align*}
\end{theorem}

In the static case, the non-classical propagators defined above do not depend on~$\tau$.
They are the natural propagators to consider in that situation, see also our earlier work~\cite{derezinski-siemssen}.

In the non-static case, however, the instantaneous non-classical propagators just defined have deficiencies from the physical point of view, see \eg~\cite{fulling}.
First of all, their definition hinges on the arbitrary choice of a fixed instance of time and, even more seriously, on the choice of a time function.
Secondly, instantaneous positive frequency bisolutions usually do not satisfy the microlocal spectrum condition of~\cite{radzikowski} (in other words, they do not define Hadamard states).

Nevertheless, the situation improves if the Klein--Gordon operator is infinitesimally static at the time when the positive/negative frequency splitting is performed.
In a forthcoming article~\cite{derezinski-siemssen:hadamard} we will show (using methods of evolution equations) that the corresponding instantaneous positive frequency bisolutions, which we define in the following section, satisfy then the microlocal spectrum condition of~\cite{radzikowski}.

\section{Asymptotic non-classical propagators}
\label{sec:asymptotic}

Throughout this section we assume that Assumption~\ref{asm:asymptotic} is satisfied.
It follows, in particular, that $B(t)$ converges to $B(\pm\infty)$ as $t \to \pm\infty$ in norm as an operator from $\Hen$ to $\Hne$.
We define the \emph{out} and \emph{in positive/negative frequency projections}
\begin{align*}
  \Pi_+^{(\pm)} &\defn \one_{[0,\infty[}\bigl(\pm B(+\infty)\bigr),\\
  \Pi_-^{(\pm)} &\defn \one_{[0,\infty[}\bigl(\pm B(-\infty)\bigr).
\end{align*}

\begin{theorem}
  The strong limits
  \begin{subequations}\begin{align}
    \Pi^{(\pm)}_+(t) &\defn \slim_{\tau \to +\infty} U(t,\tau ) \Pi^{(\pm)}_{+} U(\tau ,t), \label{eq:prop1} \\
    \Pi^{(\pm)}_-(t) &\defn \slim_{\tau \to -\infty} U(t,\tau ) \Pi^{(\pm)}_{-} U(\tau,t) \label{eq:prop2}
  \end{align}\end{subequations}
  exist as bounded operators on $\mathcal{H}_\lambda$ with $\lambda \in [-1,1]$.
  They satisfy the obvious analogs of Propositions \ref{prps1} and \ref{prop:Q_pos}.
Besides,
  \begin{align}
    U(s,t) \Pi^{(\pm)}_+(t) U(t,s) &= \Pi^{(\pm)}_+(s), \label{eq:pp1+} \\
    U(s,t) \Pi^{(\pm)}_-(t) U(t,s) &= \Pi^{(\pm)}_-(s). \label{eq:pp1-}
  \end{align}
\end{theorem}
\begin{proof}
  We only prove the theorem for~\eqref{eq:prop1} because the proof for~\eqref{eq:prop2} is the same.
  We have
  \begin{equation*}
    U(t,r) \Pi^{(\pm)}_+ U(r,t) = U(t,r) \e^{\im(t-r) B(+\infty)} \Pi^{(\pm)}_+ \e^{\im(r-t) B(+\infty)} U(r,t).
  \end{equation*}
  We analyze separately the limit $r \to +\infty$ of the operators left and right of the projection.
  Since both operators are bounded on $\mathcal{H}_{\lambda,\tau}$, $\lambda \in [-1,1]$, uniformly in $t,r$ for arbitrary $\tau \in \RR$, it is sufficient to show the convergence on $\Hen$ with respect to the norm on $\Hne[\tau]$.

  We may assume that $r > t$.
  For $u \in \Hen$ we have
  \begin{align*}
    U(t, r) \e^{\im (t-r) B(+\infty)} u
    &= u +\! \int_t^r \partial_s \bigl( U(t, s) \e^{\im (t-s) B(+\infty)} \bigr) u\, \dif s \\
    &= u - \im\! \int_t^r U(t, s) \bigl( B(s) - B(+\infty) \bigr) \e^{\im (t-s) B(+\infty)} u\, \dif s,
  \end{align*}
  by the fundamental theorem of calculus and~\ref{item:evolution:3} of Thm.~\ref{thm:evolution}.
  Taking the norm of this expression in $\Hne[\tau]$, we find
  \begin{equation*}\begin{split}
    \MoveEqLeft \norm[\big]{U(t, r) \e^{\im (t-r) B(+\infty)} u - u}_{\en^*,\tau} \\
    &\leq C \norm{u}_{\en,\tau} \int_t^r \norm[\big]{\bigl( \one \oplus L(\tau)^{-\frac12} \bigr) \bigl( B(s) - B(+\infty) \bigr) \bigl( L(\tau)^{-\frac12} \oplus \one \bigr)}\, \dif s,
  \end{split}\end{equation*}
  since $U(t,s)$ is uniformly bounded on $\Hne[\tau]$.

  It follows from the proof of Prop.~\ref{prop:norm_cont} that
  \begin{equation*}
    \norm[\big]{\bigl( \one \oplus L(\tau)^{-\frac12} \bigr) \bigl( B(s) - B(+\infty) \bigr) \bigl( L(\tau)^{-\frac12} \oplus \one \bigr)}
  \end{equation*}
  is uniformly bounded.
  Therefore,
  \begin{equation*}
    \norm[\big]{U(t, r) \e^{\im (t-r) B(+\infty)} u - u}_{\en^*,\tau} \to 0
  \end{equation*}
  as $t, r \to +\infty$ and the desired convergence follows.

  The proof for $U(t,r) \e^{\im (t-r) B(+\infty)}$ is essentially the same.
  The main difference is that we use the uniform boundedness of $U(t,s)$ on $\Hen[\tau]$.
\end{proof}

We also define
\begin{align}
  E_+^{(\pm)}(t,s) &\defn \pm U(t,\tau)\Pi_+^{(\pm)}(\tau)U(\tau,s),\\
  E_-^{(\pm)}(t,s) &\defn \pm U(t,\tau)\Pi_-^{(\pm)}(\tau)U(\tau,s).
\end{align}
Clearly, the definition above do not depend on $\tau$.

The kernels $E^{(\pm)}_\pm(t,s)$ yield the \emph{positive/negative frequency bisolutions at future} and \emph{past infinity}. They are often called \emph{out} and \emph{in}, or jointly \emph{asymptotic}.
Moreover, we may use them together with the advanced and retarded propagators to define corresponding \emph{asymptotic Feynman} and \emph{anti-Feynman propagators}:
\begin{align*}
  E^\Feyn_\pm  &= E^\wedge + E^{(+)}_\pm = E^\vee   + E^{(-)}_\pm, \\
  E^\aFeyn_\pm &= E^\vee   - E^{(+)}_\pm = E^\wedge - E^{(-)}_\pm.
\end{align*}
As before, the propagators $E^\bullet_\pm$ for $\partial_t + \im B$ induce the corresponding propagators $G^\bullet_\pm$ for $\KG$.
Obviously, the asymptotic non-classical propagators defined here have analogues to Thm.~\ref{thm:E_quantum} and Thm.~\ref{thm:G_quantum}; we only have to replace occurrences of~$\tau$ with~$\pm$.

The asymptotic propagators defined above have various advantages over the instantaneous ones of the previous section.
For one,
they do not depend on an arbitrarily chosen instant of time.
Under rather broad assumptions one can show that they even do not depend on the choice of the time function, but only on the spacetime itself.
Finally, as recently discussed in~\cite{gerard-wrochna:inout}, if the spacetime becomes asymptotically static sufficiently fast,
they satisfy the microlocal spectrum condition of~\cite{radzikowski}.

\begin{acknowledgments}
  We would like to thank Kenji Yajima for useful discussions.
  The work of D.S. was supported by a grant of the Polish National Science Center (NCN) based on the decision no. DEC-2015/16/S/ST1/00473.
  The work of J.D. was supported by the National Science Center under the grant UMO-2014/15/B/ST1/00126.
\end{acknowledgments}

\appendix
\section{Second order differential operators}
\label{appx:laplace}

Consider a manifold $\Sigma$.
Every second-order Hermitian differential operator on $L^2(\Omega^\frac12\Sigma)$ can locally be written as
\begin{equation}
  L = D_i g^{ij}(x) D_j - A^i(x) D_i - D_i A^i(x) + Y_0(x),
\end{equation}
where $g^{ij}=g^{ji}$, $Y_0$ and $A^i$ are real-valued.

$L$ can be often rewritten in the form
\begin{equation}\label{eq:L-appx}
  L = (D_i - A_i) g^{ij} (D_j - A_j) + Y_1.
\end{equation}
This is possible in particular if $g^{ij}$ is everywhere non-degenerate, \viz, $g$ determines a (pseudo\nobreakdash-)\hspace{0pt}Riemannian structure on~$M$.
Then~\eqref{eq:L-appx} holds with
\begin{equation*}
  A_i \defn g_{ij} A^j,
  \quad
  Y_1 \defn Y_0 - A^i g_{ij} A^j,
\end{equation*}
where $g_{ij}$ denotes the inverse of $g^{ij}$.

Let $\gamma$ be an everywhere non-zero function.
Then the operator~$L$ can be rewritten as
\begin{equation}\label{eq:L-pregeom-appx}
  L = \gamma^{-\frac12} (D_i - A_i) \gamma g^{ij} (D_j - A_j) \gamma^{-\frac12} + Y_\gamma,
\end{equation}
where
\begin{align*}
  Y_\gamma &\defn Y - \tfrac12 \bigl(D_i g^{ij} \gamma^{-1} (D_j \gamma)\bigr) - \tfrac14 g^{ij} \gamma^{-2} (D_i \gamma) (D_j \gamma).
\end{align*}
In particular, if we set $\gamma \defn \abs{g}^\frac12$, where $\abs{g} \defn \abs{\det [g_{ij}]}$ is the canonical density induced by the metric, and $Y \defn Y_{|g|^{\frac12}}$, then~\eqref{eq:L-pregeom-appx} yields the geometric form of the operator~$L$:
\begin{equation}\label{eq:L-geom-appx}
  L = \abs{g}^{-\frac14} (D_i-A_i) \abs{g}^\frac12 g^{ij} (D_j-A_j) \abs{g}^{-\frac14} + Y.
\end{equation}
If $g$ is a metric tensor, $A$ a $1$-form, and $Y$ a scalar, then the right-hand side of~\eqref{eq:L-geom-appx} transforms covariantly and $L$ is well-defined as a differential operator acting on half-densities.
We can rewrite~\eqref{eq:L-geom-appx} using the Levi-Civita derivative $\nabla$ for~$g$ as
\begin{equation}\label{eq:L-geom2-appx}
  L = g^{ij} (\im\nabla_{\!i}+A_i) (\im\nabla_{\!j}+A_j) + Y.
\end{equation}
Note that in~\eqref{eq:L-geom2-appx} the right $\nabla$ acts on half-densities and the left $\nabla$ acts on half-densitized covectors.

If the metric is Riemannian, the differential part of the operator~\eqref{eq:L-geom-appx} can be called a \emph{(magnetic) Laplace--Beltrami operator}, and the full operator can be called a \emph{(magnetic) Schr\"odinger operator}.
If the metric is Lorentzian, the differential part of the operator~\eqref{eq:L-geom-appx} can be called an \emph{(electromagnetic) d'Alembertian}, and the full operator can be called an \emph{(electromagnetic) Klein--Gordon operator}.

It is however sometimes convenient to consider a density $\gamma$ independent of the metric tensor $g$, \ie, to work with~\eqref{eq:L-pregeom-appx} instead of~\eqref{eq:L-geom-appx}.
Using the derivative
\begin{equation}
  D^{A,\gamma} \defn \gamma^{\frac12}(D - A)\gamma^{-\frac12},
\end{equation}
$L$ can be written as a quadratic form on half-densities:
\begin{equation}\label{eq:L-form-appx}
  (u\,|\,Lv) = \int_\Sigma \bigl( (\conj{D_i^{A,\gamma} u}) g^{ij} (D_j^{A,\gamma} v) + \conj{u}\, Y_\gamma\, v \bigr).
\end{equation}

\begin{assumption}\label{asm:laplace}
  In the remaining part of this appendix we assume that~$g$ is a Riemannian metric.
  We also assume that $\gamma^{-1} \partial_i\gamma,\ A_i \in L^2_\loc(\Sigma)$, $g^{ij} \in L_\loc^\infty(\Sigma)$ and $Y_\gamma \in L^1_\loc(\Sigma)$ such that $Y_\gamma \geq C$ for some $C \in \RR$.
\end{assumption}

We will see that under the above assumption $L$ can be understood as a self-adjoint operator on~$L^2(\Omega^\frac12\Sigma)$ in at least two natural ways.
First we reinterpret~\eqref{eq:L-form-appx} by introducing the form
\begin{equation}\label{eq:lmax}
  l_{\mx}[u,v] = \int_\Sigma \bigl( (\conj{D_i^{A,\gamma} u}) g^{ij} (D^{A,\gamma}_j v) + \conj{u}\, Y_\gamma\, v \bigr)
\end{equation}
on its maximal form domain
\begin{equation*}
  \dom l_{\mx} = \bigl\{ u \in L^2(\Omega^\frac12\Sigma) \;\big|\; D^{A,\gamma} u \in L^2(\Omega^\frac12 T^*\!\Sigma, g),\; Y_\gamma^\frac12 u \in L^2(\Omega^\frac12\Sigma) \bigr\}.
\end{equation*}
Here we denote by $L^2(\Omega^\frac12 T^*\!\Sigma, g)$ the completion of $C^\infty_\comp(\Omega^\frac12 T^*\!\Sigma)$ with respect to the norm given by
\begin{equation*}
  u \mapsto \left( \int_\Sigma \conj{u}_i \,g^{ij} u_j \right)^\frac12.
\end{equation*}
We remark that $C^\infty_\comp(\Omega^\frac12\Sigma) \subset \dom l_{\mx}$.

The following is a standard proof and has been adapted from Lem.~1 of~\cite{leinfelder-simader}.

\begin{lemma}\label{lem:selfadj-ext}
  The form $l_{\mx}$ is closed and Hermitian.
  It defines a unique self-adjoint operator $L_{\mx}$ on
  \begin{equation*}
    \Dom L_{\mx} = \bigl\{ v \in \dom l_{\mx} \;\big|\; \abs{l_{\mx}[u,v]} \leq C_v \norm{u} \;\text{for all}\; u \in L^2(\Omega^\frac12\Sigma) \bigr\}
  \end{equation*}
  satisfying
  \begin{equation*}
    (u \,|\, L_{\mx} v) = l_{\mx}[u,v]
  \end{equation*}
  for $u \in \dom l_{\mx}$ and $v \in \Dom L_{\mx}$.
  Moreover, $\dom l_{\mx} = \Dom L_{\mx}^\frac12$.
\end{lemma}
\begin{proof}
  Suppose that $\{u_n\} \subset \dom l_{\mx}$ is a Cauchy sequence with respect to the norm
  \begin{equation*}
    \dom l_{\mx} \ni u \mapsto \bigl( l_{\mx}[u,u] + (1-C)\norm{u}^2 \bigr)^\frac12.
  \end{equation*}
  Then there exist $u, v \in L^2(\Omega^\frac12\Sigma)$ and $w \in L^2(\Omega^\frac12T^*\!\Sigma, g)$ such that
  \begin{equation*}
    u_n \to u,
    \quad
    Y_\gamma^\frac12 u_n \to v
    \quad
    \text{in}\; L^2(\Omega^\frac12\Sigma)
  \end{equation*}
  and
  \begin{equation*}
    D^{A,\gamma} u_n \to w
    \quad
    \text{in}\; L^2(\Omega^\frac12 T^*\!\Sigma, g).
  \end{equation*}
  Moreover, $Y_\gamma^\frac12 u_n \to Y_\gamma^\frac12 u$ and $D^{A,\gamma}u_n \to D^{A,\gamma}u$ weakly, and thus $v = Y_\gamma^\frac12 u$ and $w = D^{A,\gamma} u$ because $v,w$ must coincide with the weak limits.
  It follows that $l_{\mx}$ is a closed form (and manifestly Hermitian).
  Therefore, by the first representation theorem (Thm.~VI.2.6 of~\cite{kato}), $l_{\mx}$ defines a unique self-adjoint operator with the stated properties.
\end{proof}

An alternative to $l_{\mx}$ is the form $l_{\mn}$ given by the completion of the form~\eqref{eq:lmax} on $C^\infty_\comp(\Omega^\frac12\Sigma)$, and the corresponding operator $L_{\mn}$.
$l_{\mn}$ may have a strictly smaller domain than $l_{\mx}$ because of boundary effects.
If $l_{\mn}=l_{\mx}$, then $C^\infty_\comp(\Omega^\frac12\Sigma)$ is a core of~$l_{\mx}$.
Note that for $\Sigma = \RR^3$ with the Euclidean metric this is known to be true, see \eg~\cite{leinfelder-simader}.

Certainly the setting considered in this appendix is not the most general possible.
For example, the assumption that $Y$ is bounded from below can certainly be relaxed.

\section{Concrete assumptions}
\label{appx:assumptions}

The objective of this appendix is to eludicate how Assumption~\ref{asm:LW_growth_loc} may be realized in practice.
Recall that $(\Sigma, \tilde{g}_\Sigma(t))$ is a family of Riemannian manifolds, $\gamma(t) > 0$ are densities on~$\Sigma$, $A(t)$ are real-valued $1$-forms and $\tilde{Y}(t)$ are real-valued scalar potentials.
For simplicity, we write $\tilde g$ for $\tilde g_\Sigma$.
As in Assumption~\ref{asm:laplace} in Appx.~\ref{appx:laplace}, we assume that $\gamma^{-1}(t)\partial_i\gamma(t),\ A_i(t) \in L^2_\loc(\Sigma)$, $\tilde{g}^{ij} \in L_\loc^\infty(\Sigma)$, and $\tilde{Y} \in L^1_\loc(\Sigma)$ is bounded from below.

Let us recall the definition of the operators $W(t)$ and $L(t)$ on $L^2(\Omega^\frac12\Sigma)$:
\begin{align}
  W(t) &\defn \beta(t)^i D_i + V(t) - \frac12 \gamma(t)^{-1} \bigl(D_t \gamma(t) - \beta(t)^i D_i \gamma(t)\bigr), \nonumber \\
  (u \,|\, L(t)\, v) &\defn \int_\Sigma \Bigl( \bigl(\conj{D^{A,\gamma}_i(t)\, u}\bigr) \tilde{g}^{ij}(t) \bigl(D^{A,\gamma}_j(t)\, v\bigr) + \conj{u}\,\tilde{Y}(t)\, v \Bigr), \label{eq:L1}
\end{align}
where $L(t)$ is interpreted, say, as the maximal operator given by~\eqref{eq:L1}, as in Appx.~\ref{appx:laplace}.
Assumption~\ref{asm:LW_growth_loc} now says that there exists a positive $C \in L^1_\loc(\RR)$ such that for all $\abs{t-s} \leq 1$
\begin{equation}\label{eq:LW-growth-loc-appx}
  \norm[\big]{L(t)^{-\frac12} \bigl( L(t) - L(s) \bigr) L(t)^{-\frac12}} + 2 \norm[\big]{\bigl( W(t) - W(s) \bigr) L(t)^{-\frac12}} \leq \abs*{\int_s^t C(r)\, \dif r},
\end{equation}
for some $C \in L^1_\loc(\RR)$.

We also introduce the family of norms
\begin{equation*}
  \norm{X}_t = \biggl(\int_\Sigma \tilde{g}^{ij}(t)\, \conj{X}_i X_j \biggr)^\frac12
\end{equation*}
for half-densitized $1$-forms $X$ on $\Sigma$.

\begin{proposition}\label{prop:concrete-asm-1}
  Suppose that there are positive $C_Y, C_g, C_W \in L^1_\loc(\RR)$, $C_A, C_\gamma \in L^2_\loc(\RR)$ such that for all $\abs{t-s} \leq 1$
  \begin{align*}
    \norm[\big]{L(t)^{-\frac12} \partial_s \tilde{Y}(s) L(t)^{-\frac12}} &\leq C_Y(s), \\
    \norm[\big]{\partial_s W(s) L(t)^{-\frac12}} &\leq C_W(s), \\
    \norm[\big]{\partial_s A(s) L(t)^{-\frac12}}_t &\leq C_A(s), \\
    \norm[\big]{\partial_s \gamma(s)^{-1} \dif\gamma(s) L(t)^{-\frac12}}_t &\leq C_\gamma(s), \\
    \abs[\big]{\partial_s \tilde{g}^{ij}(s) X_i X_j} &\leq C_g(s) \tilde{g}^{ij}(t) X_i X_j, \mathrlap{\quad X \in C(T^*\!\Sigma).}
  \end{align*}
  Then \eqref{eq:LW-growth-loc-appx} holds and thus Assumption~\ref{asm:LW_growth_loc} is true.
\end{proposition}
\begin{proof}
  To avoid notational clutter within this proof, we simply write $D_i$ for $D^{A,\gamma}_i$.
  Clearly, the assumptions of the proposition imply
  \begin{subequations}\label{eq:diff-ineqs}\begin{align}
    \norm[\big]{L(t)^{-\frac12} \bigl( \tilde{Y}(t) - \tilde{Y}(s) \bigr) L(t)^{-\frac12}} &\leq \abs*{\int_s^t C_Y(r)\, \dif r}, \\
    \norm[\big]{\bigl( W(t) - W(s) \bigr) L(t)^{-\frac12}} &\leq \abs*{\int_s^t C_W(r)\, \dif r}, \\
    \norm[\big]{\bigl( A(t) - A(s) \bigr) L(t)^{-\frac12}}_t &\leq \abs*{\int_s^t C_A(r)\, \dif r}, \\
    \norm[\big]{\bigl( \gamma(t)^{-1} \dif\gamma(t) - \gamma(s)^{-1} \dif\gamma(s) \bigr) L(t)^{-\frac12}}_t &\leq \abs*{\int_s^t C_\gamma(r)\, \dif r}, \\
    \abs[\big]{\tilde{g}^{ij}(t) X_i X_j - \tilde{g}^{ij}(s) X_i X_j} &\leq \abs*{\int_s^t C_g(r)\, \dif r}\, \tilde{g}^{ij}(t) X_i X_j.
  \end{align}\end{subequations}

  We compute
  \begin{align*}\MoveEqLeft
    \bigl(u \,\big|\, \bigl( L(t) - L(s) \bigr) u\bigr) \\
    &= \int_\Sigma \tilde{g}^{ij}(t) \Bigl( \bigl( \conj{D_i(t) u} \bigr) \bigl( D_j(t) u - D_j(s) u \bigr) + \bigl( \conj{D_i(t) u - D_i(s) u} \bigr) \bigl( D_j(t) u \bigr) \\&\qquad - \bigl( \conj{D_i(t) u - D_i(s) u} \bigr) \bigl( D_j(t) u - D_j(s) u \bigr) \Bigr) \\&\quad
    + \int_\Sigma \bigl( \tilde{g}^{ij}(t) - \tilde{g}^{ij}(s) \bigr) \Bigl( \bigl( \conj{D_i(t) u} \bigr) \bigl( D_j(t) u \bigr) - \bigl( \conj{D_i(t) u} \bigr) \bigl( D_j(t) u - D_j(s) u \bigr) \\&\qquad - \bigl( \conj{D_i(t) u - D_i(s) u} \bigr) \bigl( D_j(t) u \bigr) + \bigl( \conj{D_i(t) u - D_i(s) u} \bigr) \bigl( D_j(t) u - D_j(s) u \bigr) \Bigr) \\&\quad
    + \int_\Sigma \bigl( \tilde{Y}(t) - \tilde{Y}(s) \bigr) \abs{u}^2,
  \end{align*}
  where
  \begin{equation*}
    D_i(t) - D_i(s) = -A_i(t) + A_i(s) + \frac\im2 \gamma(t)^{-1} \partial_i \gamma(t) - \frac\im2 \gamma(s)^{-1} \partial_i \gamma(s).
  \end{equation*}
  Estimating each term separately using~\eqref{eq:diff-ineqs}, we find
  \begin{equation*}
    \abs*{\bigl(u \,\big|\, \bigl( L(t) - L(s) \bigr) u\bigr)} \leq \tilde{C}(t,s) (u \,|\, L(t) u),
  \end{equation*}
  where
  \begin{equation*}\begin{split}
    \tilde{C}(t,s) &= 2 \abs*{\int_s^t C_D(r)\, \dif r} + \abs*{\int_s^t C_D(r)\, \dif r}^2 \\&\quad + \abs*{\int_s^t C_g(r)\, \dif r} \left( 1 + \abs*{\int_s^t C_D(r)\, \dif r} \right)^2 + \abs*{\int_s^t C_Y(r)\, \dif r}
  \end{split}\end{equation*}
  with $C_D = C_A + C_\gamma / 2$.
  After two applications of
  \begin{equation*}
    \abs*{\int_s^t C_D(r)\, \dif r}^2 \leq \abs{t-s} \abs*{\int_s^t C_D(r)^2\, \dif r} \leq \abs*{\int_s^t C_D(r)^2\, \dif r},
  \end{equation*}
  which is a simple consequence of the Cauchy--Schwarz inequality, we obtain
  \begin{equation*}
    \tilde{C}(t,s) \leq \abs*{\int_s^t \bigl( c(t) (2C_D + C_D^2) + C_\gamma + C_g \bigr)\, \dif r},
  \end{equation*}
  where $c(t): = 1 + \int_{t-1}^{t+1} C_g(r)\, \dif r$.
  Thus Assumption~\ref{asm:LW_growth_loc} is true with $C(t)=\tilde{C}(t)+C_W(t).$
\end{proof}

The inequalities~\eqref{eq:diff-ineqs} in the last proposition were stated with respect to~$L(t)$.
For a more convenient criterion, fix a (time-independent) Riemannian metric~$g_0$ on~$\Sigma$ and set $\gamma_0 \defn \abs{g_0}^\frac12$.
Consider the operator $L_0$ defined by the form
\begin{equation*}
  (u \,|\, L_0 v) \defn \int_\Sigma \bigl( (\conj{D^{\gamma_0}_i u}) g_0^{ij}(t) (D^{\gamma_0}_j v) + \conj{u}\, v \bigr).
\end{equation*}

\begin{proposition}
  Assume that there exists a positive $C_g \in C(\RR)$ such that
  \begin{equation}\label{eq:gammaY-below0}
    \tilde{g}^{ij}(t) X_i X_j \geq C_g(t) g_0^{ij} X_i X_j.
  \end{equation}
  Further, suppose that there exist $\varepsilon_0 \in C(\RR)$, $\varepsilon_0(t) \in \mathopen{]}0,1\mathclose{[}$, and a positive $C_0 \in C(\RR)$
  such that
  \begin{equation}\label{eq:gammaY-below}
    \varepsilon_0(t) \gamma_0^2 \gamma(t)^{-2} \bigl(\partial_i \gamma_0^{-1} \gamma(t)\bigr) \tilde{g}^{ij}(t) \bigl(\partial_j \gamma_0^{-1} \gamma(t) \bigr) + \tilde{Y}(t) \geq C_0(t)
  \end{equation}
  Then there exists a positive $C \in C(\RR)$ such that $L_0$ satisfies the inequality
  \begin{equation}\label{eq:tildeL-ineq}
    \norm[\big]{L(t)^\frac12 u} \geq C(t) \norm[\big]{L_0^\frac12 \abs{u}},
    \quad
    u \in \Dom L(t)^\frac12.
  \end{equation}
\end{proposition}
\begin{proof}
  Let $\varepsilon(t) \defn (1-4\varepsilon_0(t))^{-1}$, so that $\varepsilon_0(t) = \frac14(1-\varepsilon(t)^{-1})$.
  Then
  \begin{align*}
    (u \,|\, L(t) u)
    &\geq \int_\Sigma \Bigl( -\bigl(D^\gamma_i(t) \abs{u}\bigr) \tilde{g}^{ij}(t) \bigl(D^\gamma_j(t) \abs{u}\bigr) + \tilde{Y}(t)\, \abs{u}^2 \Bigr) \\
    &\geq \int_\Sigma \bigl(\varepsilon(t)-1\bigr) (D^{\gamma_0}_i \abs{u}) \tilde{g}^{ij}(t) (D^{\gamma_0}_j \abs{u}) \\&\qquad + \int_\Sigma \bigl( \varepsilon_0(t) \gamma_0^2 \gamma(t)^{-2} \bigl(\partial_i \gamma_0^{-1} \gamma(t)\bigr) \tilde{g}^{ij}(t) \bigl(\partial_j \gamma_0^{-1} \gamma(t) \bigr) + \tilde{Y}(t) \bigl) \abs{u}^2 \\
    &\geq \min\bigl(C_g(t) (1-\varepsilon(t)), C_0(t)\bigr)\, \bigl(\abs{u} \,\big|\, L_0 \abs{u}\bigr).
  \end{align*}
  In the first step we used the diamagnetic inequality
  \begin{equation*}
    \abs[\big]{\bigl( \partial_x - \im V(x) \bigr) f(x)} \geq \abs[\big]{\partial_x \abs{f(x)}}
  \end{equation*}
  almost everywhere for real $V$ and $f$ such that $(\partial_x - \im V) f$ exists almost everywhere.
  In the second step we used the Cauchy--Schwarz inequality.
\end{proof}

We can apply the preceding proposition to restate Prop.~\ref{prop:concrete-asm-1} using $L_0$ instead of $L(t)$.
For this purpose we introduce another norm on half-densitized $1$-forms:
\begin{equation*}
  \norm{X} = \biggl(\int_\Sigma g_0^{ij}\, \conj{X}_i X_j \biggr)^\frac12.
\end{equation*}

\begin{proposition}
  In addition to~\eqref{eq:gammaY-below0} and~\eqref{eq:gammaY-below} we suppose that for some $C_g \in C(\RR)$
  \begin{equation*}
    \tilde{g}^{ij}(t) X_i X_j \leq C_g(t) g_0^{ij} X_i X_j, \mathrlap{\quad X \in C(T^*\!\Sigma).}
  \end{equation*}
  Moreover, we assume that there are positive $ C_{Y,0}, C_{g,0}, C_{W,0} \in L^1_\loc(\RR)$, $C_{A,0}, C_{\gamma,0} \in L^2_\loc(\RR)$ such that for all $t\in\RR$
  \begin{align*}
    \norm[\big]{L_0^{-\frac12} \abs{\partial_t \tilde{Y}(t)} L_0^{-\frac12}} &\leq C_{Y,0}(t), \\
    \norm[\big]{\partial_t W(t)L_0^{-\frac12}} &\leq C_{W,0}(t), \\
    \norm[\big]{\partial_t A(t) L_0^{-\frac12}} &\leq C_{A,0}(t), \\
    \norm[\big]{\partial_t \gamma(t)^{-1} \dif\gamma(t) L_0^{-\frac12}} &\leq C_{\gamma,0}(t), \\
    \abs[\big]{\partial_t \tilde{g}^{ij}(t) X_i X_j} &\leq C_{g,0}(t) g_0^{ij} X_i X_j, \mathrlap{\quad X \in C(T^*\!\Sigma).}
  \end{align*}
  Then Assumption~\ref{asm:LW_growth_loc} is true.
\end{proposition}

\section{Non-autonomous evolution equations}
\label{appx:evolution}

To make this paper more self-contained, we explain in this appendix relevant aspects of the theory of linear evolution equations.
We are more general than strictly necessary for the purposes of this paper, but in anticipation of our upcoming work this generality could be useful.
The results stated in this appendix can be found in similar form in~\cite{kato:hyperbolic} and in the monographs~\cite{pazy,tanabe}.
We also wish to refer to the appendix of the recent work~\cite{bach-bru} by Bach and Bru, which uses slightly different assumptions that essentially coincide with ours for the Hilbertian case.
Finally, we would like to mention~\cite{schmid-griesemer} which also discusses the theory of non-autonomous evolution equation on uniformly convex Banach spaces.

Let $\mathcal{X}$ be a Banach space.
We recall that a linear operator $A$ on $\mathcal{X}$ is the generator of a strongly continuous (one-parameter) semigroup $[0,\infty[ \ni t \mapsto \e^{tA}$ if and only if $A$ is densely defined, closed and there exist constants $M \geq 1$, $\beta \in \RR$ such that its resolvent satisfies
\begin{equation}\label{eq:group}
  \norm{(A - \lambda)^{-n}} \leq M (\lambda - \beta)^{-n},
  \quad \lambda > \beta,
  \quad n = 1,2,\dotsc.
\end{equation}
Then we have $\norm{\e^{tA}} \leq M \e^{\beta t}$ and say that $\e^{tA}$ is a semigroup of type $(M,\beta)$.
If both $A$ and $-A$ generate strongly continuous semigroups, they generate a strongly continuous (one-parameter) group $\RR \ni t \mapsto \e^{tA}$.

If \begin{equation}\label{eq:group1}
  \norm{(A - \lambda)^{-1}} \leq (\lambda - \beta)^{-1},
  \quad \lambda > \beta,
\end{equation}
then~\eqref{eq:group} is true with $M=1$.
Then $\norm{\e^{tA}} \leq M \e^{\beta t}$, so that $\e^{t(A-\beta)}$ is a semigroup of contractions.

Let $\mathcal{Y}$ be another Banach space, which is densely and continuously embedded in $\mathcal{X}$.
\begin{definition}
  By the \emph{part of $A$ on} $\mathcal{Y}$ we mean the operator $\tilde{A}$, which is the restriction of $A$ to the domain
  \begin{equation*}
    \Dom(\tilde{A}) \defn \{y\in\Dom(A)\cap\mathcal{Y}\mid Ay\in\mathcal{Y}\}.
  \end{equation*}
\end{definition}

\begin{definition}
  $\mathcal{Y}$ is called \emph{$A$-admissible} if the semigroup $\e^{tA}$, $t \in [0,\infty[$, leaves $\mathcal{Y}$ invariant and its restriction to $\mathcal{Y}$ is a strongly continuous semigroup on $\mathcal{Y}$.
\end{definition}

In the following we consider a family $\{A(t)\}_{t \in [0,T]}$ of generators of a strongly continuous semigroup.
We chose the interval $[0,T]$ for convenience and definiteness; the generalization to other intervals is straightforward.

\begin{definition}
  The family $\{A(t)\}_{t \in [0,T]}$ is called \emph{stable} with stability constants $M \geq 1$, $\beta \in \RR$, if
  \begin{equation*}
    \norm[\bigg]{\prod_{j=1}^k \bigl( A(t_j) - \lambda \bigr)^{-1}} \leq M (\lambda - \beta)^{-k},
    \quad \lambda > \beta,
  \end{equation*}
  for all finite sequences $0 \leq t_1 \leq t_2 \leq \dotsb \leq t_k \leq T$, $k=1,2,\dotsc$.
  Here and below such products are time-ordered (\viz, factors with a larger $t_j$ are to the left of factors with a smaller~$t_j$).
\end{definition}

\begin{proposition}\label{prop:stability-e}
  If $\{A(t)\}_{t \in [0,T]}$ is stable with stability constants $M, \beta$, then
  \begin{equation*}
    \norm[\bigg]{\prod_{j=1}^k \e^{\mu_j A(t_j)}} \leq M \e^{\beta (\mu_1 + \dotsb + \mu_k)},
    \quad \mu_j \geq 0.
  \end{equation*}
\end{proposition}
\begin{proof}
  The proof is straightforward, see \eg\ Prop.~7.3 of~\cite{tanabe}.
\end{proof}

The following simple generalization of Prop.~3.4 in~\cite{kato:hyperbolic} gives a criterion for the stability uses an assumption of the form~\eqref{eq:group1} for a time-dependent norm:

\begin{proposition}\label{prop:equivalent-norms}
  For each $t \in [0,T]$, let $\norm{\,\cdot\,}_t$ be an equivalent norm on $\mathcal{X}$ and $C \in L^1[0,T]$ positive such that
  \begin{equation}\label{eq:equivalent-norm}
    \norm{u}_s \leq \norm{u}_t \exp\,\abs*{\int_s^t C(r)\, \dif r},
    \quad u \in \mathcal{X},
    \quad s,t \in [0,T].
  \end{equation}
  If $\{A(t)\}$ satisfies
  \begin{equation}\label{eq:almost-contraction}
    \norm[\big]{\bigl( A(t) - \lambda \bigr)^{-1}}_t \leq (\lambda - \beta)^{-1},
    \quad \lambda > \beta,
  \end{equation}
  for all $t \in [0,T]$, then for any $s\in[0,T]$
  \begin{equation*}\label{eq:special-stability}
    \norm[\bigg]{\prod_{j=1}^k \bigl( A(t_j) - \lambda \bigr)^{-1}}_s \leq (\lambda - \beta)^{-k} \exp\left(\int_0^T 2C(r)\, \dif r\right),
    \quad t_1 \leq s \leq t_k,
  \end{equation*}
  for every finite sequence $0 \leq t_1 \leq t_2 \leq \dotsb \leq t_k \leq T$.
\end{proposition}
\begin{proof}
  Repeated application of~\eqref{eq:equivalent-norm} and~\eqref{eq:almost-contraction} yields
  \begin{align*}
    \norm[\bigg]{\prod_{j=1}^k \bigl( A(t_j) - \lambda \bigr)^{-1} u}_{t_k}
    &\leq (\lambda - \beta)^{-1} \norm[\bigg]{\prod_{j=1}^{k-1} \bigl( A(t_j) - \lambda \bigr)^{-1} u}_{t_k} \\
    &\leq (\lambda - \beta)^{-1} \exp\left(\int_{t_{k-1}}^{t_k} C(r)\, \dif r\right) \norm[\bigg]{\prod_{j=1}^{k-1} \bigl( A(t_j) - \lambda \bigr)^{-1} u}_{t_{k-1}} \\
    &\leq \dotsb \\
    &\leq (\lambda - \beta)^{-k} \exp\left(\int_{t_1}^{t_k} C(r)\, \dif r\right) \norm{u}_{t_1}.
  \end{align*}
  Applying~\eqref{eq:equivalent-norm} twice more (for $s$ and $t_k$, as well as $s$ and $t_1$), we obtain the desired result.
\end{proof}

Let us start with a rather general theorem on the construction of evolution operators (see also Thm.~4.1 of~\cite{kato:hyperbolic} and Thm.~7.1 of~\cite{tanabe}).
Note that the properties of the evolution operator describe in this theorem are rather modest.
\begin{theorem}\label{thm:evolution1}
  Assume that:
  \begin{enumerate}[label=(\alph*)]
    \item $\{A(t)\}_{t \in [0,T]}$ is stable with constants $M, \beta$.
    \item $\mathcal{Y}$ is $A(t)$-admissible for each $t$, and the part $\tilde{A}(t)$ of $A(t)$ in $\mathcal{Y}$ is stable with constants $\tilde{M}, \tilde\beta$.
    \item $\mathcal{Y} \subset \Dom A(t)$ so that $A(t) \in B(\mathcal{Y},\mathcal{X})$ for each $t$, and $t \mapsto A(t)$ is norm-continuous in the norm of $B(\mathcal{Y},\mathcal{X})$.
  \end{enumerate}
  Then there exists a unique family of bounded operators $\{U(t,s)\}_{0 \leq s \leq t \leq T}$, on $\mathcal{X}$, called the \emph{evolution (operator)} generated by $A(t)$, with the following properties:
  \begin{enumerate}
    \item \label{item:evolution-appx:1}
      For all $0 \leq r \leq s \leq t \leq T$, we have the identities
      \begin{equation*}
        U(t,t) = \one,
        \quad
        U(t,s) U(s,r) = U(t,r).
      \end{equation*}
    \item \label{item:evolution-appx:2}
      $(t,s) \mapsto U(t,s)$ is strongly $\mathcal{X}$-continuous and $\norm{U(t,s)}_{\mathcal{X}} \leq M \e^{\beta (t-s)}$.
    \item \label{item:evolution-appx:3}
      For all $y \in \mathcal{Y}$ and $0 \leq s \leq t \leq T$,
      \begin{subequations}\begin{align}
        \partial_t^+ U(t,s) y \big|_{t=s} &= A(s) y, \label{eq:U_right_deriv_t} \\
        -\partial_s U(t,s) y &= U(t,s) A(s) y, \label{eq:U_deriv_s}
      \end{align}\end{subequations}
      where the right derivative $\partial_t^+$ and the derivative $\partial_s$ (right derivative if $s=0$ and left derivative if $s=t$) are in the strong topology of $\mathcal{X}$.
  \end{enumerate}
\end{theorem}
\begin{proof}
  We approximate $A(t)$ by step functions:
  Set
  \begin{equation*}
    A_n(t) = A(T\floor{tn/T}/n),
  \end{equation*}
  where $\floor{\,\cdot\,}$ denotes the floor function, \viz, rounding to the integral part.
  Since $t \mapsto A(t)$ is norm-continuous in the norm of $B(\mathcal{Y},\mathcal{X})$, we have
  \begin{equation}\label{eq:lim_An_to_A}
    \norm{A_n(t) - A(t)}_{B(\mathcal{Y},\mathcal{X})} \to 0
    \quad\text{as}\quad
    n \to \infty
  \end{equation}
  uniformly in $t$.
  It follows immediately that also $A_n(t)$ and $\tilde{A}_n(t)$ are stable with constants $M, \beta$ and $\tilde{M}, \tilde\beta$, respectively.

  Corresponding to $A_n(t)$ we construct approximating evolution operators $U_n(t,s)$ by setting
  \begin{equation*}
    U_n(t,s) = \e^{(t-s) A_n(s)}
  \end{equation*}
  if $s,t$ belong to the closure of an interval where $A_n$ is constant, and by imposing the relation
  \begin{equation*}
    U_n(t,s) = U_n(t,r) U_n(r,s),
  \end{equation*}
  to determine $U_n(t,s)$ for other values of $s,t$.
  Clearly, $U_n(t,t) = \one$ and $(t,s) \mapsto U_n(t,s)$ is strongly $\mathcal{X}$-continuous.
  We also have
  \begin{equation}\label{eq:Un-bounds}
    \norm{U_n(t,s)}_\mathcal{X} \leq M \e^{\beta (t-s)},
    \quad
    \norm{U_n(t,s)}_\mathcal{Y} \leq \tilde{M} \e^{\tilde\beta (t-s)}
  \end{equation}
  by Prop.~\ref{prop:stability-e}, and $U_n(t,s) \mathcal{Y} \subset \mathcal{Y}$ because $\mathcal{Y}$ is $A(t)$-admissible.
  Furthermore, because $\mathcal{Y} \subset \Dom A(t)$ we have for $y \in \mathcal{Y}$
  \begin{align*}
    \partial_t U_n(t,s) y &= A_n(t) U_n(t,s) y, \\
    \partial_s U_n(t,s) y &= -U_n(t,s) A_n(s) y,
  \end{align*}
  for any $t$ resp. $s$ that is not on the boundary of an interval where $A_n$ is constant.

  Next we show that $U_n(t,s)$ converges to $U(t,s)$ strongly in $\mathcal{X}$ uniformly in $s,t$:
  By the fundamental theorem of calculus, we have
  \begin{equation*}
    U_n(t,r) y - U_m(t,r) y
    = \int_r^t U_n(t,s) \bigl( A_n(s) - A_m(s) \bigr) U_m(s,r) y\, \dif s.
  \end{equation*}
  Applying~\eqref{eq:Un-bounds}, we thus obtain
  \begin{equation*}
    \norm{U_n(t,r) y - U_m(t,r) y}_{\mathcal{X}} \leq M \tilde{M} \e^{\gamma (t-r)} \norm{y}_{\mathcal{Y}} \int_r^t \norm{A_n(s) - A_m(s)}_{B(\mathcal{Y},\mathcal{X})}\, \dif s,
  \end{equation*}
  where $\gamma = \max(\beta,\tilde\beta)$.
  Therefore it follows from~\eqref{eq:lim_An_to_A} that $U_n(t,s) y$ converges in the strong topology of $\mathcal{X}$ uniformly in $s,t$.
  Since $\mathcal{Y}$ is dense in $\mathcal{X}$ and $U_n(t,s)$ is uniformly bounded in $n$, $U_n(t,s)$ converges strongly in $\mathcal{X}$ and we set
  \begin{equation*}
    U(t,s) = \slim_{n\to\infty} U_n(t,s).
  \end{equation*}
  It is immediate that the properties~\ref{item:evolution-appx:1} and~\ref{item:evolution-appx:2} follow from the corresponding properties for $U_n(t,s)$.

  Finally, we show uniqueness and~\ref{item:evolution-appx:3}:
  If $\{V(t,s)\}_{0 \leq s \leq t \leq T}$ satisfies \ref{item:evolution-appx:1}--\ref{item:evolution-appx:3} for a stable family of operators $\{A'(t)\}_{t \in [0,T]}$ with the same stability constants, then we apply the fundamental theorem of calculus to find
  \begin{equation*}
    U_n(t,s) y - V(t,s) y
    = \int_s^t U_n(t,r) \bigl( A_n(r) - A'(r) \bigr) V(r,s) y\, \dif r,
  \end{equation*}
  and therefore
  \begin{equation}\label{eq:Un-V_bound}
    \norm{U_n(t,s) y - V(t,s) y}_{\mathcal{X}} \leq M \tilde{M} \e^{\gamma (t-s)} \norm{y}_{\mathcal{Y}} \int_s^t \norm{A_n(r) - A'(r)}_{B(\mathcal{Y},\mathcal{X})}\, \dif r.
  \end{equation}
  If we set $A'(t) = A(t)$ and let $n \to \infty$, we thus find that $U(t,s) y = V(t,s) y$ and by density $U(t,s) = V(t,s)$ on the whole of $\mathcal{X}$.
  We conclude that $U(t,s)$ is unique.

  Now, in~\eqref{eq:Un-V_bound}, we set $A'(t) = A(\tau) = \mathrm{const}$ for $\tau \in [0,T]$, divide by $t-s$ and let $n \to \infty$ to obtain
  \begin{equation*}
    (t-s)^{-1} \norm{U(t,s) y - \e^{(t-s) A(\tau)} y}_\mathcal{X} \leq (t-s)^{-1} M \tilde{M} \e^{\gamma (t-s)} \norm{y}_{\mathcal{Y}} \int_s^t \norm{A_n(r) - A(\tau)}_{B(\mathcal{Y},\mathcal{X})}\, \dif r.
  \end{equation*}
  On the one hand, for $\tau = s$, we find~\eqref{eq:U_right_deriv_t} in the limit $t \to s$.
  On the other hand, setting $\tau = t$ and letting $t \to s$, we find
  \begin{equation}\label{eq:U_left_deriv_s}
    \partial_s^- U(t,s) y \big|_{s=t} = -A(t) y.
  \end{equation}
  To find~\eqref{eq:U_deriv_s}, we check the right and left derivative separately.
  Applying~\eqref{eq:U_right_deriv_t} and~\eqref{eq:U_left_deriv_s}, we obtain
  \begin{subequations}\label{eq:U_deriv_s_proof}\begin{align}
    \begin{split}
      \partial_s^+ U(t,s) y
      &= \slim_{h \searrow 0} h^{-1} \bigl( U(t,s+h) y - U(t,s) y \bigr) \\
      &=  U(t,s+h) \slim_{h \searrow 0} h^{-1} \bigl( y - U(s+h,s) y \bigr)
       = -U(t,s) A(s) y,
    \end{split} \\
    \begin{split}
      \partial_s^- U(t,s) y
      &= \slim_{h \searrow 0} h^{-1} \bigl( U(t,s) y - U(t,s-h) y \bigr) \\
      &=  U(t,s) \slim_{h \searrow 0} h^{-1} \bigl( y - U(s,s-h) y \bigr)
       = -U(t,s) A(s) y.
    \end{split}
  \end{align}\end{subequations}
  Therefore we have completed the proof also for~\ref{item:evolution-appx:3}.
\end{proof}

We say that a Banach space $\mathcal{Y}$ possesses a predual if there exists a Banach space $\mathcal{Y}_*$ such that $\mathcal{Y}$ is the dual of $\mathcal{Y}_*$.
Having fixed a predual $\mathcal{Y}_*$, we can equip $\mathcal{Y}$ with the so-called weak* topology, which is generated by the seminorms $y\mapsto|\xi(y)|$, where $\xi\in\mathcal{Y}_*$.
Note in particular that every reflexive Banach space possesses a unique predual (namely, its dual).
For reflexive Banach spaces the weak* convergence clearly coincides with the weak convergence.

For Banach spaces possessing a predual one can slightly improve the previous theorem (see also Thm.~5.1 of~\cite{kato:hyperbolic}).
\begin{theorem}\label{thm:evolution_reflexive}
  In addition to the assumptions of Thm.~\ref{thm:evolution1}, assume that:
  \begin{enumerate}
    \item[(d)] $\mathcal{Y}$ possesses a predual.
  \end{enumerate}
  Then, in addition to~\ref{item:evolution-appx:1}--\ref{item:evolution-appx:3}, the evolution $\{U(t,s)\}_{0 \leq s \leq t \leq T}$ has the following properties:
  \begin{enumerate}[start=4]
    \item \label{item:evolution-appx:4}
      $U(t,s) \mathcal{Y} \subset \mathcal{Y}$, $(t,s) \mapsto U(t,s)$ is weakly* continuous and
      \begin{equation}\label{eq:group2}
        \norm{U(t,r)}_{\mathcal{Y}} \leq \tilde{M} \e^{\tilde\beta (t-s)},
        \quad 0 \leq r \leq s \leq t \leq T.
      \end{equation}
  \end{enumerate}
\end{theorem}
\begin{proof}
  Note that for fixed $s,t \in [0,T]$ and $y \in \mathcal{Y}$, $U_n(t,s) y$ is a uniformly bounded sequence in~$\mathcal{Y}$, and thus, by the Banach--Alaoglu Theorem, it contains a weakly* convergent subsequence.
  Moreover, by our previous results, $U_n(t,s) y \to U(t,s) y$ in $\mathcal{X}$.
  But $U(t,s) y$ must be equal to the weak* limit, and thus lie in $\mathcal{Y}$, \ie, $U(t,s) \mathcal{Y} \subset \mathcal{Y}$.
  The inequality then follows from~\eqref{eq:Un-bounds}.

  Now, let $(t_j)_j$, $(s_j)_j$ be sequences with $t_j \to t$, $s_j \to s$ and $y \in \mathcal{Y}$.
  Recall that $U(t_j,s_j) y \to U(t,s) y$ in $\mathcal{X}$ because $(t,s) \mapsto U(t,s)$ is $\mathcal{X}$-strongly continuous.
  By the Banach--Alaoglu Theorem, since $U(t_j,s_j) $ is uniformly bounded, $U(t_j,s_j) y$ contains a weakly* convergent subsequence.
  The weak* limit of $U(t_j,s_j) y$ is thus $U(t,s) y$ and must lie in~$\mathcal{Y}$.
  In other words, $U(t,s)$ is weakly* continuous on~$\mathcal{Y}$.
\end{proof}

We recall that a normed space is called \emph{uniformly convex} if for every $\varepsilon > 0$ and unit vectors $\norm{x} = \norm{y} = 1$ there exists $\delta > 0$ such that
\begin{equation*}
  \norm{x-y} \geq \varepsilon
  \quad\Rightarrow\quad
  \norm*{\frac{x+y}{2}} \leq 1 -\delta.
\end{equation*}

We note that all uniformly convex Banach spaces are reflexive.
Besides, on uniformly convex Banach spaces $\norm{x_n} \to \norm{x}$ and the weak convergence $x_n \rightharpoonup x$ implies the strong convergence.
All Hilbert spaces are uniformly convex.

If we assume that the Banach space $\mathcal{Y}$ is uniformly convex, stronger results about the evolution can be derived.
They are described in the following theorem, which is a part of Thm.~5.2 of~\cite{kato:hyperbolic}:
\begin{theorem}\label{thm:evolution2}
  In addition to the assumptions of Thm.~\ref{thm:evolution1}, assume that
  \begin{enumerate}
    \item[(d')]
      $\mathcal{Y}$ is uniformly convex.
    \item[(e)]
      For every $t$ there exist on $\mathcal{Y}$ an equivalent norm $\norm{\,\cdot\,}_{\mathcal{Y},t}$ as well as a positive $C \in L^1[0,T]$ such that
      \begin{equation}\label{eq:exp_norm_growth}
        \norm{y}_{\mathcal{Y},s} \leq \norm{y}_{\mathcal{Y},t} \exp\,\abs*{\int_s^t C(r)\, \dif r},
        \quad s,t \in [0,T].
      \end{equation}
      Besides, there exists $\tilde\beta \in \RR$ such that
      \begin{equation*}
        \norm[\big]{\bigl( \tilde{A}(t) - \lambda \bigr)^{-1}}_{\mathcal{Y},t} \leq (\lambda - \tilde\beta)^{-1},
        \quad \lambda > \tilde\beta,
      \end{equation*}
      for all $t \in [0,T]$.
  \end{enumerate}
  Then, in addition to \ref{item:evolution-appx:1}--\ref{item:evolution-appx:3}, the evolution $\{U(t,s)\}_{0 \leq s \leq t \leq T}$ has the following property, which is an improved version of \ref{item:evolution-appx:4}:
  \begin{enumerate}
    \item[\manlabel{item:evolution-appx:4'}{(iv')}]
      $U(t,s)$ preserves $\mathcal{Y}$, is $\mathcal{Y}$-strongly continuous in $s$ for fixed $t$ and $\mathcal{Y}$-strongly right-continuous in $t$ for fixed~$s$, and
      \begin{equation}\label{eq:U-bound-Y}
        \norm{U(t,r)}_{\mathcal{Y},s} \leq \exp\left(\int_r^t \bigl(\tilde\beta + 2 C(\tau)\bigr)\, \dif\tau\right),
        \quad 0 \leq r \leq s \leq t \leq T.
      \end{equation}
  \end{enumerate}
\end{theorem}
\begin{proof}
  Since $\mathcal{Y}$ is uniformly convex, it is also reflexive, and thus~\ref{item:evolution-appx:4} of Thm.~\ref{thm:evolution_reflexive} holds.
  Then we use Props.~\ref{prop:stability-e} and~\ref{prop:equivalent-norms} to find~\eqref{eq:U-bound-Y}.

  Let us prove the strong continuity.
  By~\ref{item:evolution-appx:4}, for any $y \in \mathcal{Y}$ we have $\wlim_{t,r \to s} U(t,r) y \to y$.
  Using this, and then the bound~\eqref{eq:U-bound-Y}, we obtain
  \begin{align*}
    \norm{y} &
    \leq \liminf_{r,t \to s}\, \norm{U(t,r) y}_{\mathcal{Y},s}
    \leq \limsup_{r,t \to s}\, \norm{U(t,r) y}_{\mathcal{Y},s} \\&
    \leq \limsup_{r,t \to s}\, \exp\left(\int_r^t \bigl(\tilde\beta + 2 C(\tau)\bigr)\, \dif\tau\right) \norm{y}
    = \norm{y}.
  \end{align*}
  Hence, $\lim_{r,t \to s}\, \norm{U(t,r) y}_{\mathcal{Y},s} = \norm{y}$.
  But $\mathcal{Y}$ is uniformly convex, so this implies that
  \begin{equation*}
    \lim_{r,t \to s} U(t,r) y = y.
  \end{equation*}

  Let $0 \leq s \leq s' \leq t \leq T$ and $y \in \mathcal{Y}$.
  Then
  \begin{equation*}
    \norm{U(t,s') y - U(t,s) y}_{\mathcal{Y}} \leq \norm{U(t,s')}_{\mathcal{Y}} \norm{y - U(s',s) y}_{\mathcal{Y}} \to 0
  \end{equation*}
  as $s' \to s$ or $s \to s'$.
  Similarly, for $0 \leq s \leq t \leq t' \leq T$ we find
  \begin{equation*}
    \norm{U(t',s) y - U(t,s) y}_{\mathcal{Y}} \leq \norm{(U(t',t) - \one) U(t,s) y}_{\mathcal{Y}} \to 0
  \end{equation*}
  as $t' \to t$.
\end{proof}

In the previous theorem we still had to distinguish between between the $t$- and $s$-properties of~$U(t,s)$.
If the reversed operator $-A(T-t)$ also satisfies the assumptions of the theorems above, this distinction can be dropped, see also Remark~5.3 in~\cite{kato:hyperbolic}:

\begin{theorem}\label{thm:evolution3}
  Suppose that both $\{A(t)\}_{t \in [0,T]}$ and the reversed family $\{-A(T-t)\}_{t \in [0,T]}$ satisfy the assumptions of Thms.~\ref{thm:evolution1} and~\ref{thm:evolution2}.
  Then the unique family of bounded operators $\{U(t,s)\}_{s,t \in \RR}$ described in the previous theorems satisfies the following improved versions of (i), (iii) and (iv'):
  \begin{enumerate}[label=(\roman*')]
    \item \label{item:evolution-appx:1a}
      For all $r,s,t \in [0,T]$, we have the identities
      \begin{equation*}
        U(t,t) = \one,
        \quad
        U(t,s) U(s,r) = U(t,r).
      \end{equation*}
    \addtocounter{enumi}{1}
    \item \label{item:evolution-appx:3a}
      For all $y \in \mathcal{Y}$ and $s,t \in [0,T]$,
      \begin{subequations}\begin{align}
         \partial_t U(t,s) y &= A(t) U(t,s) y,\\
        -\partial_s U(t,s) y &= U(t,s) A(s) y,
      \end{align}\end{subequations}
      where the derivatives (right/left derivatives at the boundaries of~$[0,T]$) are in the strong topology of~$\mathcal{X}$.
    \item[(iv'')]
      $(t,s) \mapsto U(t,s)$ preserves $\mathcal{Y}$, is $\mathcal{Y}$-strongly continuous and satisfies \eqref{eq:U-bound-Y}.
  \end{enumerate}
\end{theorem}
\begin{proof}
  Denote the evolution for $\{A(t)\}_{t \in [0,T]}$ by $U(t,s)$ and the evolution for $\{-A(T-t)\}_{t \in [0,T]}$ by $V(t,s)$.
  For $0 \leq s \leq t \leq T$, we define
  \begin{equation*}
    U(s,t) = V(T-s, T-t).
  \end{equation*}
  From the approximations $U_n(t,s)$ and $V_n(t,s)$, it is easy to see that
  \begin{equation*}
    U(t,s) U(s,t) = \one
  \end{equation*}
  for $s,t \in \RR$.
  This proves~\ref{item:evolution-appx:1a}.

  It is clear that
  \begin{align*}
     \partial_t U(t,s) y \big|_{t=s} &= A(s) y,\\
    -\partial_s U(t,s) y &= U(t,s) A(s) y
  \end{align*}
  for $s,t \in [0,T]$.
  Then we can proceed as in~\eqref{eq:U_deriv_s_proof} to find also
  \begin{equation*}
    \partial_t U(t,s) y = A(t) U(t,s) y.
  \end{equation*}

  Finally, the strong continuity of $U(t,s)$ follows from~\ref{item:evolution-appx:4'} applied to both $U(t,s)$ and $V(t,s)$, which implies, in particular, that $U(t,s)$ is strongly right- and left-continuous in~$t$ for fixed~$s$.
\end{proof}

Theorem~\ref{thm:evolution3} implies the following, see also Thm.~3.2 of~\cite{yajima}:
\begin{theorem}\label{thm:evolution_selfadjoint}
  Let~$\mathcal{X}$ and~$\mathcal{Y}$ be Hilbert spaces such that $\mathcal{Y}$ is densely and continuously embedded in~$\mathcal{X}$.
  Let $I \subset \RR$ be a compact interval, and $\{A(t)\}_{t \in I}$ a family of densely defined, closed operators on~$\mathcal{X}$.
  Suppose that the following is satisfied:
  \begin{enumerate}[label=(\alph*)]
    \item $\mathcal{Y} \subset \Dom A(t)$ so that $A(t) \in B(\mathcal{Y}, \mathcal{X})$ and $t \mapsto A(t)$ is norm-continuous in the norm of~$B(\mathcal{Y},\mathcal{X})$. \label{item:evolution-selfadj:a}
    \item \label{item:evolution-selfadj:b}
      For every $t \in I$, there exist on $\mathcal{X}$ and $\mathcal{Y}$ Hilbert structures $(\,\cdot\;|\;\cdot\,)_{\mathcal{X},t}$ and $(\,\cdot\;|\;\cdot\,)_{\mathcal{Y},t}$, which are equivalent to the original ones and for a positive $C \in L^1(I)$ and all $s,t \in I$
      \begin{align*}
        \norm{x}_{\mathcal{X},s} &\leq \norm{x}_{\mathcal{X},t} \exp\,\abs*{\int_s^t C(r)\, \dif r}, \\
        \norm{y}_{\mathcal{Y},s} &\leq \norm{y}_{\mathcal{Y},t} \exp\,\abs*{\int_s^t C(r)\, \dif r}.
      \end{align*}
      Denote the corresponding Hilbert spaces~$\mathcal{X}_t$ and~$\mathcal{Y}_t$.
    \item $A(t)$ is self-adjoint with respect to~$\mathcal{X}_t$ and the part~$\tilde{A}(t)$ of~$A(t)$ in~$\mathcal{Y}_t$ is self-adjoint in~$\mathcal{Y}_t$. \label{item:evolution-selfadj:c}
  \end{enumerate}
  Then there exists a unique family of bounded operators $\{U(t,s)\}_{s,t \in I}$, in~$\mathcal{X}$, called the \emph{evolution (operator)} generated by~$A(t)$, with the following properties:
  \begin{enumerate}[label=(\roman*)]
    \item \label{item:evolution-selfadj:1}
      For all $r,s,t \in I$, we have the identities
      \begin{equation*}
        U(t,t) = \one,
        \quad
        U(t,s) U(s,r) = U(t,r).
      \end{equation*}
    \item \label{item:evolution-selfadj:2}
      $U(t,s)$ is $\mathcal{X}$-strongly continuous and
      \begin{equation*}
        \norm{U(t,s)}_{\mathcal{X},s} \leq \exp\,\abs*{\int_s^t 2 C(r)\, \dif r},
        \quad s,t \in I.
      \end{equation*}
    \item \label{item:evolution-selfadj:3}
      For all $y \in \mathcal{Y}$ and $s,t \in I$,
      \begin{align*}
         \im\partial_t U(t,s) y &= A(t) U(t,s) y,\\
        -\im\partial_s U(t,s) y &= U(t,s) A(s) y,
      \end{align*}
      where the derivatives (right/left derivatives at the boundaries of~$I$) are in the strong topology of~$\mathcal{X}$.
    \item \label{item:evolution-selfadj:4}
      $U(t,s) \mathcal{Y} \subset \mathcal{Y}$, $U(t,s)$ is $\mathcal{Y}$-strongly continuous and
      \begin{equation*}
        \norm{U(t,s)}_{\mathcal{Y},s} \leq \exp\,\abs*{\int_s^t 2 C(r)\, \dif r},
        \quad s,t \in I.
      \end{equation*}
  \end{enumerate}
\end{theorem}

The following perturbation theorem is essentially Thm.~4.5 of~\cite{kato}.
We leave the proof as an exercise to the reader.
\begin{theorem}\label{thm:perturbed_evolution}
  Suppose that $\{A(t)\}_{t \in [0,T]}$ satisfies the assumptions of Thm.~\ref{thm:evolution1}.
  Let $\{B(t)\}_{t \in [0,T]}$ be a family of bounded operators in~$\mathcal{X}$ such that $t \mapsto B(t)$ is strongly continuous with respect to $\mathcal{X}$ and $K = \sup_t \norm{B(t)}_{\mathcal{X}}$.
  Then there exists a unique evolution $V(t,s)$ for $\{A(t) + B(t)\}_{t \in [0,T]}$ satisfying the properties~\ref{item:evolution-appx:1}--\ref{item:evolution-appx:3}, but with the estimate $\norm{V(t,s)} \leq M \e^{(\beta + K M)(t-s)}$.

  Suppose that $\{A(t)\}_{t \in [0,T]}$ also satisfies the stronger assumptions of Thm.~\ref{thm:evolution_reflexive}, Thm.~\ref{thm:evolution2} or Thm.~\ref{thm:evolution3}, and $\{B(t)\}_{t \in [0,T]}$ preserves $\mathcal{Y}$ and its part $\{\tilde B(t)\}_{t \in [0,T]}$ in $\mathcal{Y}$ is bounded in~$\mathcal{Y}$ with $\tilde K = \sup_t \norm{B(t)}_{\mathcal{Y}}$.
  Then the evolution $V(t,s)$ satisfies the corresponding stronger properties, where the estimate~\eqref{eq:group2} needs to be multiplied by $\e^{\tilde K \tilde M (t-s)}$, and the estimate~\eqref{eq:U-bound-Y} by $\e^{\tilde K (t-s)}$.
\end{theorem}

The evolution $V(t,s)$ in the theorem above is given symbolically by
\begin{equation*}
  V = U + U * B * U + U * B * U * B * U + \dotsb,
\end{equation*}
where $\mathop{*} B \mathop{*}$ denotes a Volterra-type convolution with `density' $B(t)$.
For example,
\begin{equation*}
  (U * B * U)(t,r) = \int_r^t U(t,s) B(s) U(s,r)\, \dif s.
\end{equation*}

\section{Heinz--Kato inequality}
\label{appx:interpolation}

We recall the Heinz--Kato inequality \cite{heinz:heinz-kato,kato:heinz-kato}, which is an elementary but very useful result for the interpolation of operators:
\begin{theorem}\label{thm:heinz-kato}
  Suppose that $A, B$ are positive operators on Hilbert spaces $\mathcal{X}$, $\mathcal{Y}$, respectively.
  If $T$ is a bounded operator from~$\mathcal{X}$ to~$\mathcal{Y}$ such that $T (\Dom A) \subset \Dom B$ and
  \begin{equation*}
    \norm{T x} \leq C_0 \norm{x},
    \quad
    \norm{B T x} \leq C_1 \norm{A x},
  \end{equation*}
  for $x \in \Dom A$, then
  \begin{equation}\label{eq:heinz-kato}
    \norm{B^\lambda T x} \leq C_0^\lambda C_1^{1-\lambda} \norm{A^\lambda x},
    \quad
    \lambda \in [0,1].
  \end{equation}
\end{theorem}

\section{Finite speed of propagation}
\label{appx:finite-speed}

In this appendix we prove the finite speed of propagation for solutions of the Klein--Gordon equation with coefficients of low regularity.

In this section we prefer to work with the Klein--Gordon equation in the scalar formalism, given by~\eqref{eq:klein-gordon}, which can be locally written as
\begin{equation}\label{eq:klein-gordon-scalar}
  \KG u \defn -g^{\mu\nu} (\nabla_\mu - \im A_\mu) (\nabla_\nu - \im A_\nu) u + Y u
\end{equation}
with pseudo-Riemannian metric $g$ and the corresponding Levi-Civita derivative $\nabla$, vector potential $A$, and scalar potential $Y$.
Our standing assumptions in this appendix are as follows:
\begin{assumption}\label{asm:finite-speed}
  $M = \RR \times \Sigma$ is equipped with a continuous Lorentzian metric $g = -\alpha^2\, \dif t^2 + g_\Sigma$, where $\alpha > 0$ and $g_\Sigma$ are continuous, and $g_\Sigma$ restricts to a family of Riemannian metrics on $\Sigma$.
  (Recall that every globally hyperbolic spacetime can be brought into this form.)
  We assume that $A_\mu(t) \in L^\infty_\loc(\Sigma)$ for all $t$, and $A_\mu, \dot{A}_\mu, Y \in L^\infty_\loc(M)$.
  Moreover, in every compact neighbourhood $U \subset M$ there is $C_g > 0$ such that
  \begin{equation*}
    \abs{\dot{g}^{\mu\nu} X_\mu X_\nu} \leq C_g \abs{g^{\mu\nu} X_\mu X_\nu}
  \end{equation*}
  almost everywhere in~$U$ for all covectors $X$.
\end{assumption}

Under these assumption we will show the following thorem on the finite speed of propagation:
\begin{theorem}\label{thm:finite-speed}
  If $u \in C^1(\RR; L^2_\loc(\Sigma))$ with $\partial_i u \in C(\RR; L^2_\loc(\Sigma))$ and $\KG u \in L^2_\loc(M)$, then
  \begin{equation*}
    \supp u \subset J\Bigl( \supp \KG u \cup \{t\}{\times}\bigl(\supp u(t)\cup \supp \dot{u}(t)\bigr)\Bigr)
  \end{equation*}
  for any $t \in \RR$.
  That is, $u$ is supported in the causal shadow of the union of $\KG u$ and of the support of its Cauchy data on $\{t\}\times\Sigma$.
\end{theorem}

\eqref{eq:klein-gordon-scalar} can be obtained via the Euler--Lagrange equations from the \emph{Lagrangian density}
\begin{equation*}
  \mathcal{L}[u] \defn -\abs{g}^\frac12 \Bigl( \bigl((\partial_\mu + \im A_\mu) \conj{u}\bigr) g^{\mu\nu} \bigl((\partial_\nu - \im A_\nu) u\bigr) + Y \abs{u}^2 \Bigr).
\end{equation*}
To the Lagrangian density $\mathcal{L}$ we can associate the \emph{momentum flux density}
\begin{equation*}
  \mathcal{P}^\mu[u] \defn -\delta^\mu_0 \mathcal{L}[u] + \frac{\partial\mathcal{L}[u]}{\partial(\partial_\mu \conj{u})} \partial_t \conj{u} + \frac{\partial\mathcal{L}[u]}{\partial(\partial_\mu u)} \partial_t u.
\end{equation*}
If the action for $\mathcal{L}$ is invariant under infinitesimal time-translations, Noether's theorem says that the momentum flux is conserved.
If the action is not time-translation invariant, $\mathcal{P}$ is in general not conserved, but it is still a useful quantity.

The energy density $\mathcal{E}=\mathcal{P}^0$ obtained from $\mathcal{L}$ is not necessarily positive.
Therefore, for technical reasons it will be convenient to replace $\mathcal{L}$ with the modified Lagrangian density
\begin{equation*}
  \tilde{\mathcal{L}}[u] \defn -\abs{g}^\frac12 \Bigl( \bigl((\partial_\mu + \im A_\mu) \conj{u}\bigr) g^{\mu\nu} \bigl((\partial_\nu - \im A_\nu) u\bigr) - (1 + \alpha^{-2} A_0^2) \abs{u}^2 \Bigr),
\end{equation*}
denoting the corresponding momentum flux density $\tilde{\mathcal{P}}$.
Using the special form of the metric, we find the energy density
\begin{align*}
  \tilde{\mathcal{E}}[u] \defn \tilde{\mathcal{P}}^0[u]
  &= \abs{g}^\frac12 \Bigl( \alpha^{-2} \abs{\dot u}^2 + \bigl((\partial_i + \im A_i) \conj{u}\bigr) g_\Sigma^{ij} \bigl((\partial_j -\im A_j) u\bigr) + \abs{u}^2 \Bigr)
\end{align*}
and the spatial momentum flux density
\begin{align*}
  \tilde{\mathcal{P}}^i[u] = \mathcal{P}^i[u]
  &= -\abs{g}^\frac12 \Bigl( \dot{\conj{u}} g_\Sigma^{ij} \bigl((\partial_j - \im A_j) u\bigr) + \dot{u} g_\Sigma^{ij} \bigl((\partial_j + \im A_j) \conj{u}\bigl) \Bigr)
\end{align*}

Below we will integrate $\partial_\mu \tilde{\mathcal{P}}^\mu$ over a region which is delimited by two constant-time surfaces and the backward lightcone of a point as described in Fig.~\ref{fig:energy-ineq}.
To rewrite this integral as an integral over the boundary of said region via Stokes' theorem, it is useful to assume that $\partial J^\pm_g(\Omega)$ is a Lipschitz topological hypersurface, see Thm.~3.9 of~\cite{beem}.
Here we denoted by $J^\pm_g(\Omega)$ the causal future ($+$) or causal past ($-$) of $\Omega$, \ie, the set of points which can be reached from $\Omega$ by future- resp. past-directed causal curves with respect to the metric $g$.
Moreover, we write $J_g(\Omega) = J^+_g(\Omega) \cup J^-_g(\Omega)$.

If $g$ is not smooth (or at least $C^2$), it is not guaranteed that $\partial J^\pm_g(\Omega)$ is a Lipschitz topological hypersurface.
However, we can approximate $g$ by smooth metrics:

If a Lorentzian metric~$\hat{g}$ has strictly larger lightcones than~$g$, \ie, each non-vanishing $g$-causal vector $X^\mu$ ($g_{\mu\nu} X^\mu X^\nu \leq 0$) is $\hat{g}$-timelike ($\hat{g}_{\mu\nu} X^\mu X^\nu < 0$), then we write
\begin{equation*}
  \hat{g} \succ g.
\end{equation*}
As shown in Prop.~1.2 of~\cite{chrusciel}, there always exists a \emph{smooth} Lorentzian metric~$\hat{g}$ with strictly larger lightcones which approximates~$g$ arbitrarily well.

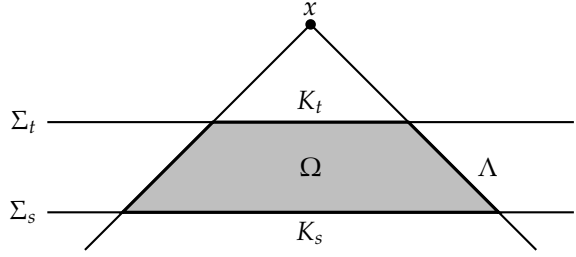
\begin{figure}
  \centering
  \begin{tikzpicture}
    \draw[fill=gray!50,very thick] (-1.3,1.2) -- (1.3,1.2) -- (2.5,0) -- (-2.5,0) -- cycle;
    \node at (0,.6) {$\Omega$};
    \draw[thick] (-3.5,0) -- (3.5,0);
    \node[left] at (-3.5,0) {$\Sigma_s$};
    \node[below] at (0,0) {$K_s$};
    \draw[thick] (-3.5,1.2) -- (3.5,1.2);
    \node[left] at (-3.5,1.2) {$\Sigma_t$};
    \node[above] at (0,1.2) {$K_t$};
    \draw[thick] (-3,-.5) -- (0,2.5) -- (3,-.5);
    \node[right] at (2.1,.6) {$\Lambda$};
    \fill (0,2.5) circle (.07);
    \node[above] at (0,2.5) {$x$};
  \end{tikzpicture}
  \caption{\label{fig:energy-ineq}%
  The truncated cone given by the backward lightcone $J^-_{\hat g}(x)$ of a point, and two constant-time surfaces $\Sigma_t = \{t\} \times \Sigma$ and $\Sigma_s$ (with $t > s$).
  We write $K_t = J^-_{\hat g}(x) \cap (\{t\} \times \Sigma)$ and $K_s$ for the caps, and $\Lambda = \partial J^-_{\hat g}(x) \cap ([s,t] \times \Sigma)$ for the mantle of the truncated cone $\Omega = J^-_{\hat g}(x) \cap ([s,t] \times \Sigma)$.
  }
\end{figure}

\begin{proposition}\label{prop:energy-ineq}
  Let $\hat{g} \succ g$ be smooth and consider the situation depicted in Fig.~\ref{fig:energy-ineq}.
  Then there exists $C > 0$ such that
  \begin{equation}\label{eq:energy_ineq}
    \e^{C(s-t)}\! \int_{K_t} \tilde{\mathcal{E}}[u](t) \leq \int_{K_s} \tilde{\mathcal{E}}[u](s) + \int_\Omega \abs{g}^\frac12 \abs{\KG u}^2.
  \end{equation}
  for all $u \in C^1(\RR; L^2_\loc(\Sigma))$ with $\partial_i u \in C(\RR; L^2_\loc(\Sigma))$ and $\KG u \in L^2_\loc(M)$,
\end{proposition}
\begin{proof}
  We derive
  \begin{align*}
    \partial_\mu \tilde{\mathcal{P}}^\mu[u]
    &= -\partial_t \tilde{\mathcal{L}}[u] + \biggl( \partial_\mu \frac{\partial\tilde{\mathcal{L}}[u]}{\partial(\partial_\mu \conj{u})} \biggr) \dot{\conj{u}} + \frac{\partial\tilde{\mathcal{L}}}{\partial(\partial_\mu \conj{u})} \partial_\mu \partial_t \conj{u} + \biggl( \partial_\mu \frac{\partial\tilde{\mathcal{L}}}{\partial(\partial_\mu u)} \biggr) \dot{u} + \frac{\partial\tilde{\mathcal{L}}[u]}{\partial(\partial_\mu u)} \partial_\mu \partial_t u \\
    &= -\partial_t \tilde{\mathcal{L}}[u] + \biggl( \abs{g}^\frac12 \tilde\KG u + \frac{\partial\tilde{\mathcal{L}}[u]}{\partial\conj{u}} \biggr) \dot{\conj{u}} + \frac{\partial\tilde{\mathcal{L}}[u]}{\partial(\partial_\mu \conj{u})} \partial_t \partial_\mu \conj{u} + \biggl( \abs{g}^\frac12 \conj{\tilde\KG u} + \frac{\partial\tilde{\mathcal{L}}[u]}{\partial u} \biggr) \dot{u} \\&\quad + \frac{\partial\tilde{\mathcal{L}}[u]}{\partial(\partial_\mu u)} \partial_t \partial_\mu u \\
    &= -2 \abs{g}^\frac12 \Re(\dot{\conj{u}} \tilde\KG u) - \frac{\partial\tilde{\mathcal{L}}[u]}{\partial g^{\mu\nu}} \dot{g}^{\mu\nu} - \frac{\partial\tilde{\mathcal{L}}[u]}{\partial A_\mu} \dot{A}_\mu - \frac{\partial\tilde{\mathcal{L}}[u]}{\partial\abs{g}} \partial_t \abs{g} \\
    &= \abs{g}^\frac12 \Bigl(2 \Re(\dot{\conj{u}} \tilde\KG u) + \bigl((\partial_\mu + \im A_\mu) \conj{u})\bigr) \dot{g}^{\mu\nu} \bigl((\partial_\nu - \im A_\nu) u\bigr) - 2 \alpha^{-3} \dot\alpha A_0^2 \abs{u}^2 \\&\qquad - 2 \Im\bigl( \conj{u} \dot{A}_\mu g^{\mu\nu} (\partial_\nu - \im A_\nu) u \bigr) + 2 \alpha^{-2} A_0 \dot{A}_0 \abs{u}^2 - \frac12 \abs{g}^{-1} (\partial_t \abs{g}) \tilde{\mathcal{L}}[u] \Bigr).
  \end{align*}
  where, in the second step, we used the Euler--Lagrange equations with
  \begin{equation*}
    \tilde\KG = \KG - Y + 1 + \alpha^{-2} A_0^2
  \end{equation*}
  being the Klein--Gordon operator associated to $\tilde{\mathcal{L}}$.
  Estimating each term separately using our assumptions and the Cauchy--Schwarz inequality yields
  \begin{equation*}
    \partial_\mu \tilde{\mathcal{P}}^\mu[u] \leq \abs{g}^\frac12 \Bigl( \abs{\KG u}^2 + C_1 \alpha^{-2} \abs{\dot{u}}^2 + C_2 \bigl((\partial_i + \im A_i) \conj{u}\bigr) g_\Sigma^{ij} \bigl((\partial_j + \im A_j) u\bigr) + C_3 \abs{u}^2 \Bigr)
  \end{equation*}
  for $C_1,C_2,C_3 > 0$ which do not depend on~$u$.
  Therefore we find
  \begin{equation}\label{eq:energy_ineq-1}
    \int_\Omega \partial_\mu \tilde{\mathcal{P}}^\mu[u] \leq \int_\Omega \bigl( \abs{g}^\frac12 \abs{\KG u}^2 + C \tilde{\mathcal{E}}[u] \bigr)
  \end{equation}
  for some constant $C > 0$.

  By Stokes' theorem,
  \begin{equation}\label{eq:energy_ineq-2}
    \int_{\Omega} \partial_\mu \tilde{\mathcal{P}}^\mu[u]
    = \int_{\partial\Omega} n_\mu \tilde{\mathcal{P}}^\mu[u]
    = \int_{K_t} \tilde{\mathcal{E}}[u](t) - \int_{K_s} \tilde{\mathcal{E}}[u](s) + \int_\Lambda n_\mu \tilde{\mathcal{P}}^\mu[u],
  \end{equation}
  where $n$ is the outward-directed normal field to~$\partial\Omega$.
  For any future-directed causal covector field~$\xi$ (\ie, $g^{\mu\nu} \xi_\mu \xi_\nu \leq 0$ and $\xi_0 \geq 0$) with $\abs{\vec\xi} = (g_\Sigma^{ij} \xi_i \xi_j)^\frac12$,
  \begin{align*}
    \xi_\mu \tilde{\mathcal{P}}^\mu[u]
    &= \xi_0 \tilde{\mathcal{E}}[u] - 2 \abs{g}^\frac12 \Re\bigl(\xi_i \dot{\conj{u}} g_\Sigma^{ij} (\partial_j - \im A_j) u\bigr) \\
    &\geq \xi_0 \tilde{\mathcal{E}}[u] - \abs{g}^\frac12 \alpha \abs{\vec\xi} \Bigl( \alpha^{-2} \abs{\dot u}^2 + \bigl((\partial_i + \im A_i u)\bigr) g_\Sigma^{ij} \bigl((\partial_j - \im A_j) u\bigr) \Bigr) \\
    &\geq (\xi_0 - \alpha \abs{\vec\xi}) \tilde{\mathcal{E}}[u]
    \geq 0
  \end{align*}
  almost everywhere.
  Consequently, we can estimate the last term in~\eqref{eq:energy_ineq-2} as $\int_\Lambda n_\mu \tilde{\mathcal{P}}^\mu \geq 0$.

  Combining~\eqref{eq:energy_ineq-1} and~\eqref{eq:energy_ineq-2}, we obtain
  \begin{equation*}
    \int_{K_t} \tilde{\mathcal{E}}[u](t) - \int_{K_s} \tilde{\mathcal{E}}[u](s) \leq \int_s^t \left( \int_{K_r} \bigl( \abs{g}^\frac12 \abs{\KG u(r)}^2 + C \tilde{\mathcal{E}}[u](r) \bigr) \right)\, \dif r,
  \end{equation*}
  and thus~\eqref{eq:energy_ineq} by Grönwall's inequality.
\end{proof}

Now, using the proposition above, we can show the finite speed of propagation:
\begin{theorem}
  If $u \in C^1(\RR; L^2_\loc(\Sigma))$ with $\partial_i u \in C(\RR; L^2_\loc(\Sigma))$ and $\KG u \in L^2_\loc(M)$, then
  \begin{align}
    \supp u \cap M_\pm &\subset J^\pm_g\Bigl( (\supp \KG u \cap M_\pm) \cup \{t\}{\times}\bigl(\supp u(t) \cup \supp \dot{u}(t) \bigr)\Bigr), \label{eq:supp} \\
    \supp u &\subset J_g\Bigl( \supp \KG u \cup \{t\}{\times}\bigl(\supp u(t) \cup \supp \dot{u}(t)\bigr)\Bigr), \notag
  \end{align}
  for any $t \in \RR$, where $M_+ = [t,+\infty\mathclose{[} \times \Sigma$, $M_- = \mathopen{]}-\infty,t] \times \Sigma$.
\end{theorem}
\begin{proof}
  Note that, as a subset of $\Sigma$, we have $\supp \tilde{\mathcal{E}}[u](t)= \supp u(t)\cup \supp \dot{u}(t)$.
  We show that $u(x) = 0$ for any
  \begin{equation*}
    x \in M \setminus J^+_{\hat g}\bigl( (\supp \KG u \cap M_+) \cup \{t\}{\times}\supp \tilde{\mathcal{E}}[u](t) \bigr)
  \end{equation*}
  by an application of Prop.~\ref{prop:energy-ineq} for all smooth $\hat{g} \succ g$.
  For any such $x$, $J^-_{\hat g}(x)$ does not intersect $(\supp \KG u \cap M_+) \cup \{t\}{\times}\supp \tilde{\mathcal{E}}[u](t)$.
  Prop.~\ref{prop:energy-ineq} now shows that $u$ vanishes in $J^-_{\hat g}(x) \cap M_+$ and thus also at~$x$.

  We have thus shown that
  \begin{equation*}
    \supp u \cap M_\pm \subset J^\pm_{\hat g}\Bigl( (\supp \KG u \cap M_\pm) \cup \{t\}{\times}\bigl(\supp u(t) \cup \supp \dot{u}(t)\bigr)\Bigr)
  \end{equation*}
  for all smooth $\hat{g} \succ g$.
  It follows that~\eqref{eq:supp} holds, because a vector is $g$-causal if and only if it is $\hat{g}$-timelike for all smooth $\hat{g} \succ g$ by Prop.~1.5 of~\cite{chrusciel} and therefore
  \begin{equation*}
    J^\pm_g(\Omega) = \bigcap_{\hat{g} \mkern2mu\succ\mkern1mu g} J^\pm_{\hat g}(\Omega),
    \quad
    \Omega \subset M.
  \end{equation*}

  The embedding for~$J^-$ follows by time reversal and remaining embedding by the union of the embeddings for~$J^+$ and~$J^-$.
\end{proof}


\end{document}